\documentclass[journal]{IEEEtran}
\usepackage{amsmath,amsfonts,amssymb,mathrsfs}
\usepackage{bbm} 
\usepackage{algorithm}
\usepackage{algpseudocode}

\usepackage{tabularx}
\usepackage{booktabs}
\usepackage{makecell}

\usepackage{cite}

\usepackage[caption=false,font=footnotesize]{subfig}
\usepackage{graphicx}

\usepackage{tikz, pgfplots}
\usetikzlibrary{shapes, arrows, arrows.meta, fit, positioning, calc, 3d}
\usepgfplotslibrary{fillbetween}
\pgfplotsset{compat=1.18} 
\tikzset{>=latex}

\newtheorem{theorem}{Theorem}
\newtheorem{lemma}{Lemma}
\newtheorem{remark}{Remark}
\newtheorem{definition}{Definition}

\newtheorem{corollary}{Corollary}

\allowdisplaybreaks
\usepackage[acronym]{glossaries}
\newcommand{\newac}{\newacronym}
\newcommand{\ac}{\gls}

\newac{rs}{RS}{randomized stationary}
\newac{iid}{\emph{i.i.d.}}{independent and identically distributed}
\newac{zoh}{ZOH}{zero-order-hold}
\newac{aoi}{AoI}{age of information}
\newac{iot}{IoT}{internet of things}
\newac{iaws}{IAWS}{independence-assumed weighted-sum}
\newac{cra}{CRA}{confidential reconstruction accuracy}
\newac{los}{LOS}{line of sight}
\newac{nlos}{NLOS}{non-line of sight}
\newac{lra}{LRA}{legitimate reconstruction advantage}
\newac{ara}{ARA}{adversarial reconstruction advantage}
\newac{tsms}{TSMS}{two-state Markov source}
\newac{utms}{UTMS}{uniform-transition Markov source}
\newac{mmse}{MMSE}{minimum mean square error estimator}

\begin{document}

\title{Legitimate Reconstruction Advantage in Wiretapped Semantics-Aware Remote Monitoring}
\author{Bowen~Li,~\IEEEmembership{Member,~IEEE},~and~Nikolaos~Pappas,~\IEEEmembership{Senior~Member,~IEEE}
\thanks{The authors are with the Department of Computer and Information Science, Link\"oping University, Link\"oping 58183, Sweden (e-mail: bowen.li@liu.se; nikolaos.pappas@liu.se). This work has been supported by ELLIIT, the European Union (6G-LEADER, 101192080).}
}
\maketitle
\begin{abstract}
This paper investigates joint accuracy and security in semantics-aware remote monitoring, where source updates intended for a legitimate receiver may also be intercepted by an eavesdropper. We define legitimate reconstruction advantage (LRA) as the joint event that the intended receiver correctly reconstructs the current source state using received updates while the eavesdropper reconstructs it incorrectly. Source evolution changes the relevance of received updates at the two receivers, so their marginal reconstruction accuracies do not, in general, determine the probability of this joint event. Under randomized stationary (RS) transmission and zero-order-hold reconstruction at both receivers, we derive an explicit joint stationary distribution and closed-form expressions for LRA probability and persistence for general finite-state ergodic Markov sources. For two-state and uniform-transition sources, we obtain algebraic expressions that enable low-complexity optimization of RS policies. Analytical and numerical results show that maximizing the probability of legitimate advantage need not maximize its persistence. For two-state sources, rare-transmission analysis reveals that, at a fixed sum of transition probabilities, greater transition asymmetry can increase persistence while reducing LRA probability. Moreover, optimizing marginal accuracy and security metrics can select policies that are suboptimal for LRA. Counterintuitively, more frequent transmission or improved legitimate-channel reliability can increase legitimate reconstruction accuracy while reducing LRA probability. Finally, a geofencing application translates prescribed LRA requirements into spatial constraints on eavesdropper locations.
\end{abstract}


\section{Introduction}
Wireless remote monitoring enables cyber-physical systems to track evolving physical processes by transmitting status updates over wireless links for remote state reconstruction and control or decision-making~\cite{BurChaLam:J18,ParColFisLu:J18,KouPap:J21}.
Owing to the broadcast nature of wireless propagation, remote monitoring may operate over a wiretap channel, where updates intended for the legitimate receiver can also be intercepted by an eavesdropper and used to infer sensitive system states, potentially facilitating subsequent attacks~\cite{LiuXiaLiLia:J12,ZhoShaChe:J24,GilYemChoNed:J25}. 
Remote monitoring under eavesdropping must therefore address accuracy and security, supporting accurate legitimate reconstruction while limiting the eavesdropper's ability to reconstruct the current source state. 
For example, market-state updates intended for an authorized recipient can support timely trading decisions, while interception can give an unauthorized party access to the same decision-relevant information.

Classical physical-layer security studies characterize secure message transmission through metrics such as secrecy capacity, achievable secrecy rate, and secrecy outage probability~\cite{Wynn:J75,MukFakHuaSwi:J14}.
In remote monitoring, receivers track an evolving source through the latest successfully delivered update, whose time-varying relevance to the current source state has motivated semantics-aware and freshness-aware communication~\cite{SunUysEliYat:J17,KosPapAng:J17,YatSunBroKau:J21} and further inspired freshness-based security~\cite{ZhaXuLanChe:J24,XuHuZhuYua:J24,ZouCheZhuYan:J26,WanCheMohPap:C24,ZhaNiHaoLi:J25,ZhaHaoDinNi:J25}.
Two common approaches are used to characterize freshness-based security.
One focuses on the age associated with the latest secure update, with update security defined according to the underlying transmission model, such as a legitimate-reception/eavesdropping-failure event or satisfaction of a prescribed physical-layer secrecy condition~\cite{ZhaXuLanChe:J24,XuHuZhuYua:J24,ZouCheZhuYan:J26}.
The other focuses on the relative information ages at the two receivers (legitimate and eavesdropper) and evaluates secrecy through their age difference, for example, by taking the positive excess of the eavesdropper’s age over that of the legitimate receiver~\cite{WanCheMohPap:C24,ZhaNiHaoLi:J25,ZhaHaoDinNi:J25}.
Both approaches characterize security via receiver asymmetry, but age alone is insufficient to determine whether legitimate reconstruction succeeds while eavesdropper reconstruction fails.
Depending on the source dynamics, a staler observation may be correct while a fresher one is not, or observations of different ages may both be correct, or both be incorrect~\cite{LuoPap:J26,MaaKriAssEph:J20}.
Thus, characterizing joint accuracy and security requires accounting for source dynamics.

Some studies characterize source dynamics using discrete-time linear time-invariant systems and define security in terms of asymmetric minimum mean-square error estimation requirements at the two receivers~\cite{LeoQueDolDey:J19,TsiGatPap:J20,CheQuLi:J26}.
These formulations account for both receivers through marginal performance objectives or constraints, such as by keeping the legitimate receiver’s error covariance bounded while driving the eavesdropper’s to infinity~\cite{TsiGatPap:J20}.
However, such marginal, receiver-wise second-order statistics do not directly characterize the simultaneous occurrence of legitimate reconstruction success and eavesdropper reconstruction failure. Jointly optimizing receiver-wise criteria does not, in general, determine the policy that best supports this joint reconstruction outcome.
Beyond protection against eavesdropping, temporal source dynamics have been explicitly incorporated into the analysis and mitigation of false-data-injection attacks~\cite{FogBonPanPas:J26}, denial-of-service attacks~\cite{KriAssSol:J25}, and timing side-channel attacks~\cite{MasChiTalZan:J26}, as well as into source-aware privacy mechanisms such as age-dependent differential privacy~\cite{ZhaWeiBerHua:J24}.
These studies exploit source dynamics in constructing attack and defense mechanisms and motivate source-aware design for secure remote monitoring.

This work addresses two fundamental questions: \emph{(i) how to characterize joint accuracy and security} under source dynamics, and \emph{(ii) how this joint characterization reshapes transmission policy design}. To this end, we investigate joint accuracy and security in semantics-aware remote monitoring under \ac{rs} policies, with both the legitimate receiver and a passive eavesdropper adopting \ac{zoh} reconstruction, and make the following three main contributions:

(1) We introduce \ac{lra}, defined as the joint event of legitimate reconstruction success and eavesdropper reconstruction failure, to jointly characterize accuracy and security in terms of probability and persistence. For an arbitrary time-homogeneous finite-state ergodic Markov source, we derive closed-form expressions for both metrics by the stationary analysis of a three-dimensional state process and establish their rational dependence on the system parameters.

(2) For two-state and uniform-transition sources, we derive algebraic expressions for the \ac{lra} probability and persistence, which enable tractable analyses of their rare-transmission limits and system-parameter optimization.
The rare-transmission limits share a common channel-induced receiver-separation factor but capture different aspects of source dynamics: state mismatch for \ac{lra} probability and mismatch-weighted state sojourns for \ac{lra} persistence; numerical results indicate that this structure extends to general Markov sources.
We further prove that optimizing the transmission probability is equivalent to polynomial root finding, reducing the complexity from $\mathcal O(N^9/\epsilon)$ to $\widetilde{\mathcal O}(\log(1/\epsilon))$,\footnote{
The notation $\widetilde{\mathcal O}(\cdot)$ suppresses polylogarithmic factors. For scalar polynomial root finding, the bit complexity is
$\mathcal O((r_F+b)r_F^2\operatorname{polylog}(r_F+b))$, where $r_F=\deg (\mathcal P_F)$ and $b=\lceil\log_2(1/\epsilon)\rceil$ is the required bit precision~\cite{Vic:J96}.} for these structured sources and to $\mathcal O(1)$ in special cases.

(3) Numerical results validate the theoretical analysis through Monte Carlo simulations and reveal that the \ac{lra} probability and persistence can favor different transmission policies.
They further demonstrate that optimizing weighted marginal objectives can select policies that are suboptimal for the joint LRA objective, potentially leading to inefficient transmission policies and inaccurate assessments of joint reconstruction performance.
The joint characterization further reveals that more frequent transmissions or a better legitimate channel need not improve \ac{lra} probability, particularly when the eavesdropping channel is strong.
Finally, a geofencing application demonstrates the practical relevance of the proposed \ac{lra} metrics by translating prescribed \ac{lra} requirements into spatial constraints on eavesdropper locations.

\section{System Model}\label{sec:system_model}
We consider a remote monitoring system over a wiretap channel, comprising a Markov source, a transmitter (Alice), a legitimate receiver (Bob), and an eavesdropper (Eve), as illustrated in Fig.~\ref{fig:system_model}. The objective is to enable accurate reconstruction at Bob while Eve’s reconstruction is incorrect. Next, we first specify the system model and then define the \ac{lra} metrics.

\subsection{Remote Monitoring System Model}\label{subsec:remote_estimation}
We consider a remote monitoring system operating in slotted time, indexed by $t\in\mathcal T\triangleq\{1,2,\ldots\}$.
The source state remains constant within each slot and evolves across slots, while the transmission decision, packet delivery, and receiver-side reconstruction are completed within the same slot.

\subsubsection{Source Process}
We consider a finite-state time-homogeneous Markov source $\{X_t\}_{t\in\mathcal T}$ with state space $\mathcal N\triangleq\{1,2,\ldots,N\}$, where $X_t$ denotes the source state in slot $t\in \mathcal T$.
Let $\mathbf Q=[p_{i,j}]_{i,j\in\mathcal N}$ denote its transition matrix, where
\begin{equation*}
    p_{i,j}
    \triangleq
    \Pr[X_{t+1}=j\mid X_t=i],
    \, i,j\in\mathcal N.
\end{equation*}
The source chain is assumed to be ergodic and admits a unique stationary distribution $ \mathbf v=[v_1,v_2,\ldots,v_N]^{\mathsf{T}}$.

\subsubsection{Transmission Policy}
At each slot $t\in\mathcal{T}$, Alice makes a transmission decision $A_t\in\{0,1\}$. If $A_t=1$, Alice samples the current source state $X_t$ and transmits the resulting update; otherwise, Alice remains idle.
In this work, we focus on \ac{rs} policies, under which $\{A_t\}_{t\in\mathcal{T}}$ is an \ac{iid} Bernoulli process independent of the source process.
Specifically,
\begin{equation*}
\Pr[A_t=1]=p_\mathrm{a},\Pr[A_t=0]=1-p_\mathrm{a},\, t\in\mathcal{T},
\end{equation*}
where $p_\mathrm{a}\in(0,1]$ denotes the transmission probability.

\begin{figure}
    \centering
    \includegraphics[width=1\linewidth]{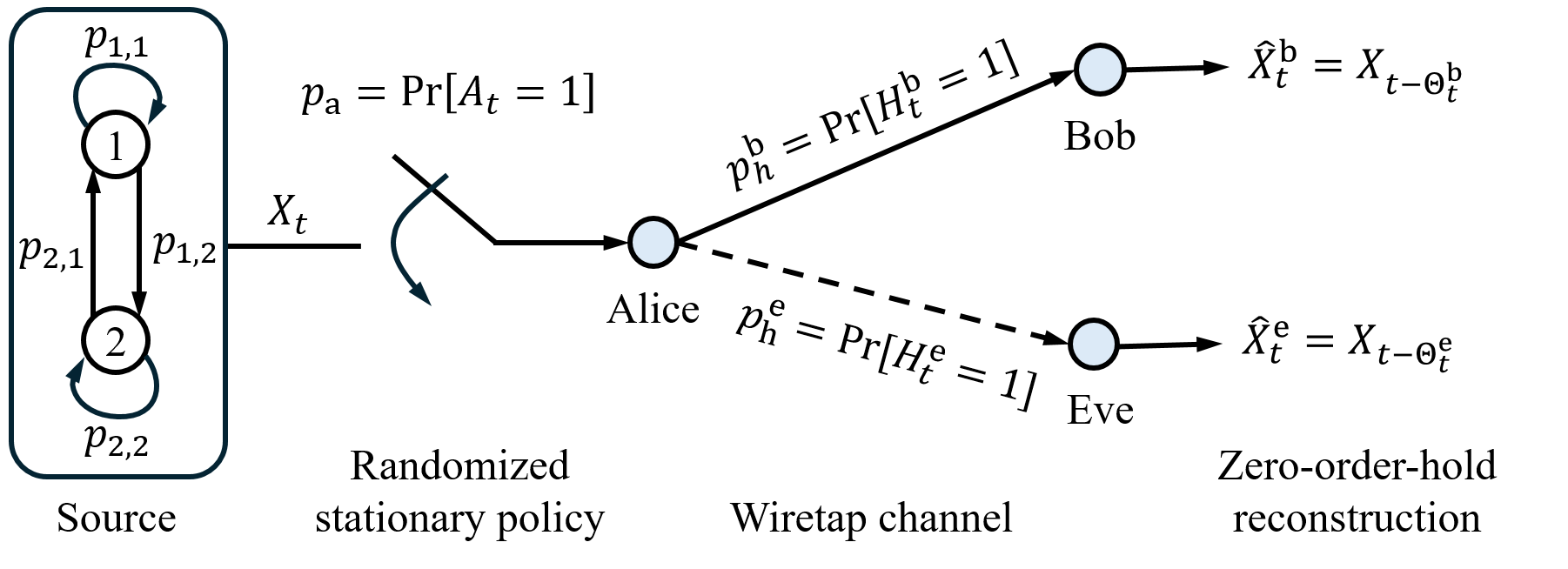}
    \caption{Remote monitoring system over a wiretap channel with a transmitter (Alice), a legitimate receiver (Bob), and an eavesdropper (Eve).}
    \label{fig:system_model}
\end{figure}

\subsubsection{Erasure Wiretap Channel}
Let $\mathcal R\triangleq\{\mathrm b,\mathrm e\}$ denote the receiver index set, where $\mathrm b$ and $\mathrm e$ correspond to Bob and Eve, respectively.
For each receiver $r\in\mathcal R$, the link from Alice to receiver $r$ is modeled as an erasure channel.
Specifically, for each slot $t\in\mathcal T$, let $H_t^r\in\{0,1\}$ denote the binary erasure-channel state of link $r$, with $H_t^r=1$ indicating successful reception of a transmitted update and $H_t^r=0$ indicating erasure.\footnote{This abstraction follows classical erasure-channel, packet-erasure, and erasure wiretap channel models~\cite{ThoJoy:B06,PhiMih:b07,SurSubThaBlo:C10}.}
The wiretap channel-state processes $\{H_t^{\mathrm b}\}_{t\in\mathcal T}$ and $\{H_t^{\mathrm e}\}_{t\in\mathcal T}$ are independent Bernoulli processes, each \ac{iid} over time and independent of both the source process and the transmission decision. Specifically,
\begin{equation*} 
\Pr[H_t^r=1]=p_\mathrm{h}^r, \, \Pr[H_t^r=0]=1-p_\mathrm{h}^r, \, r\in \mathcal R,\, t\in\mathcal T, 
\end{equation*} 
where $p_{\mathrm h}^r\in(0,1]$ denotes the transmission success probability over link $r$.

\subsubsection{Receiver Observation and Reconstruction}
Given the transmission decision and the channel state, the observation at receiver $r\in\mathcal{R}$ is
\begin{equation} 
Y_t^r= 
\begin{cases} 
X_t, & A_t H_t^r=1,\\ 
\varepsilon \notin\mathcal{N},  & A_t H_t^r=0, \end{cases} \, t\in\mathcal T.
\label{eq:def_observation}
\end{equation} 
That is, receiver $r$ observes the current source state if and only if Alice transmits and the update is successfully received over link $r$; otherwise, it observes an erasure symbol $\varepsilon$.

Each receiver adopts a \ac{zoh} reconstruction rule, retaining the latest successfully received source state until a new update arrives, particularly for receivers with limited state-inference capability.
Let $\Theta_t^r$ denote the \ac{aoi} at receiver $r$, defined as 
\begin{equation}
\Theta_t^r = t-\max\{\tau\leq t: Y_\tau^r\neq \varepsilon\}, \, r\in\mathcal{R},\, t\in\mathcal{T}. \label{eq:def_aoi_receiver} 
\end{equation}
The reconstructed state at receiver $r$ in slot $t$ is
\begin{equation} 
\hat X_t^r =Y_{t-\Theta_t^r}^r= X_{t-\Theta_t^r}, \, r\in\mathcal{R},\, t\in\mathcal{T}. \label{eq:def_estimator_receiver} 
\end{equation}
Equations \eqref{eq:def_aoi_receiver}-\eqref{eq:def_estimator_receiver} apply after the first successful reception; before then, each receiver retains an arbitrary initial reconstruction in $\mathcal N$.

\subsubsection{Adversary Model}
We consider a passive eavesdropper that observes successfully intercepted source updates without modifying the source process, transmission decisions, or channel conditions. Both Bob and Eve reconstruct the current source state using the \ac{zoh} rule specified above. Our analysis characterizes reconstruction performance under this rule, without assuming its optimality for Eve. Accordingly, security refers to reconstruction failure at Eve under this rule, and \ac{lra} captures its simultaneous occurrence with correct reconstruction at Bob.

\subsection{Legitimate Reconstruction Advantage Metrics}\label{subsec:metric_definition}
To characterize the reconstruction asymmetry between Bob and Eve in terms of probability and persistence, we define the \ac{lra} status and \ac{lra} run length as follows.
\begin{definition}
    The \ac{lra} status is defined as
    \begin{equation*}
    J_\mathrm{LRA}[t]
    \triangleq
    \mathbbm 1\{\hat X_t^\mathrm{b}=X_t,\ \hat X_t^{\mathrm e}\neq X_t\},\ t\in\mathcal T.
\end{equation*}
\end{definition}
Thus, $J_\mathrm{LRA}[t]=1$ indicates that Bob has a reconstruction advantage over Eve in slot $t$.

\begin{definition}
    The \ac{lra} run length is defined as
    \begin{align*}
        D_\mathrm{LRA}[t]\triangleq
        &\max\{d\in[0,t]\cap\mathbb{Z}:\\
        &\qquad J_\mathrm{LRA}[t-l]=1,l\in [0,d-1]\cap\mathbb{Z}\},\ t\in\mathcal T.
    \end{align*}
\end{definition}
Hence, $D_\mathrm{LRA}[t]=0$ when the \ac{lra} event does not hold in slot $t$; otherwise, $D_\mathrm{LRA}[t]$ equals the number of consecutive \ac{lra} slots ending at slot $t$.

Accordingly, the long-term average \ac{lra} probability is
\begin{equation}
    \bar J_\mathrm{LRA}
    \triangleq
    \liminf_{T\to\infty}
    \frac{1}{T}
    \sum_{t=1}^{T}
    \mathbb E[J_\mathrm{LRA}[t]].
    \label{eq:avg_reconstruction_advantage}
\end{equation}
and the long-term average \ac{lra} persistence is
\begin{equation}
    \bar D_\mathrm{LRA}
    \triangleq
    \liminf_{T\to\infty}
    \frac{1}{T}
    \sum_{t=1}^{T}
    \mathbb E[D_\mathrm{LRA}[t]].
    \label{eq:avg_reconstruction_advantage_duration}
\end{equation}

The two metrics capture complementary but distinct aspects of joint accuracy and security.
The \ac{lra} probability measures how frequently Bob has a reconstruction advantage over Eve, whereas the \ac{lra} persistence measures how long the current advantage has persisted, averaged over all slots.
Consequently, optimizing one metric does not generally optimize the other, as established analytically and further demonstrated through numerical results.

\section{General Performance Characterization}\label{sec:general_performance}
In this section, we characterize the general \ac{lra} performance. We first establish the Markovianity and ergodicity of the induced process comprising the source state and the reconstructions at Bob and Eve, and then derive its stationary distribution.
Using this joint-state characterization, we derive the closed-form expressions for the long-term \ac{lra} probability and persistence and establish their rational dependence on the system parameters.

\subsection{System Analysis}\label{subsec:system_analysis}
To capture the coupled evolution of the source and the receiver reconstructions, define the composite state $S_t \triangleq (X_t, \hat{X}^\mathrm{b}_t, \hat{X}^\mathrm{e}_t) \in \mathcal{N}^3$.
The source process, transmission decisions, and wiretap channel-state processes jointly induce the stochastic process $\{S_t\}_{t \in \mathcal{T}}$.
Based on the system model described in Section~\ref{subsec:remote_estimation}, the following result establishes the Markovianity and ergodicity of the induced process $\{S_t\}_{t \in \mathcal{T}}$.

\begin{lemma}[Markovianity]\label{lem:mc_3d}
    The induced process $\{S_t\}_{t \in \mathcal{T}}$ forms a finite-state, time-homogeneous Markov chain on $\mathcal{N}^3$. For $\mathbf{s}=(x,a,b)\in\mathcal{N}^3$ and $\mathbf{s}'=(x',a',b')\in\mathcal{N}^3$, its transition kernel is
    \begin{equation}
        [\mathbf{P}]_{\mathbf{s},\mathbf{s}'}=[\mathbf{Q}]_{x,x'}\Psi(\mathbf{s},\mathbf{s}'),
        \label{eq:kernel}
    \end{equation}
    where
    \begin{align*}
        \Psi(\mathbf{s},\mathbf{s}') &= \lambda_{11}\mathbbm{1}(a'=x',b'=x')+\lambda_{10}\mathbbm{1}(a'=x',b'=b) \nonumber\\
        & \quad +\lambda_{01}\mathbbm{1}(a'=a,b'=x')+ \lambda_{00}\mathbbm{1}(a'=a,b'=b).
    \end{align*}
    The coefficients are
\begin{align*}
    \lambda_{11} 
    &= p_{\mathrm a} p_{\mathrm h}^{\mathrm b} p_{\mathrm h}^{\mathrm e},
    & (\text{Both update}) \\
    \lambda_{10} 
    &= p_{\mathrm a} p_{\mathrm h}^{\mathrm b} (1-p_{\mathrm h}^{\mathrm e}),
    & (\text{Only Bob updates}) \\
    \lambda_{01} 
    &= p_{\mathrm a} (1-p_{\mathrm h}^{\mathrm b}) p_{\mathrm h}^{\mathrm e},
    & (\text{Only Eve updates}) \\
    \lambda_{00} 
    &= p_{\mathrm a}(1-p_{\mathrm h}^{\mathrm b})(1-p_{\mathrm h}^{\mathrm e}) + (1-p_{\mathrm a}).
    & (\text{Neither updates})
\end{align*}
\end{lemma}

\begin{IEEEproof}
    See Appendix~\ref{sec:proof_mc_3d}.
\end{IEEEproof}

\begin{lemma}[Ergodicity]\label{lem:ergodicity}
The induced process $\{S_t\}_{t\in\mathcal T}$ has a unique aperiodic recurrent class and admits a unique stationary distribution.
\end{lemma}
\begin{IEEEproof}
    See Appendix~\ref{sec:proof_ergodicity}.
\end{IEEEproof}

Lemma~\ref{lem:mc_3d} provides a time-homogeneous Markov representation of the coupled source and reconstruction dynamics, while Lemma~\ref{lem:ergodicity} ensures that its long-term behavior is independent of the initial condition.
Consequently, long-term time averages converge to the stationary expectations.
We next derive the stationary distribution induced by the transition kernel in Lemma~\ref{lem:mc_3d}, and then express the long-term \ac{lra} metrics by aggregating the stationary probabilities of the relevant composite states, as in Theorems~\ref{thm:stationary} and \ref{thm:stationary_metrics}, respectively.

For any $\mathbf{s}=(x,a,b)\in\mathcal N^3$, define the stationary distribution as
\begin{equation*}\pi(\mathbf{s}) \triangleq \lim_{t \to \infty} \Pr(X_t=x, \hat{X}^{\mathrm b}_t=a, \hat{X}^{\mathrm e}_t=b).\end{equation*}
The following theorem gives its closed-form expression.

\begin{theorem}[Stationary distribution]\label{thm:stationary}
The stationary distribution $\pi(\mathbf{s})$ of the induced process $\{S_t\}_{t \in \mathcal{T}}$ is
\begin{equation}
    \pi(\mathbf{s}) = \pi_1(\mathbf{s}) + \pi_2(\mathbf{s}) + \pi_3(\mathbf{s}),
    \label{eq:stationary_distribution}
\end{equation}
where
\begin{align*}
    \pi_1 (\mathbf{s}) &= \mathbbm 1(a=b) [\mathbf v]_a \lambda_{11} \left[ (\mathbf{I}_N - \lambda_{00} \mathbf Q)^{-1} \right]_{a,x}, \\
    \pi_2 (\mathbf{s}) &= [\mathbf v]_b \lambda_{10} p_\mathrm{e} \left[ \mathbf{Q} \left(\mathbf{I}_N - (1-p_\mathrm{e}) \mathbf Q\right)^{-1} \right]_{b,a} \nonumber \\
        &\quad \times \left[ (\mathbf{I}_N - \lambda_{00} \mathbf Q)^{-1} \right]_{a,x}, \\
    \pi_3 (\mathbf{s}) &= [\mathbf v]_a \lambda_{01} p_\mathrm{b} \left[ \mathbf{Q} \left(\mathbf{I}_N - (1-p_\mathrm{b})\mathbf{Q}\right)^{-1} \right]_{a,b} \nonumber \\
        &\quad \times \left[ (\mathbf{I}_N - \lambda_{00}\mathbf{Q})^{-1} \right]_{b,x}.
\end{align*}
Here, $p_\mathrm{b} = \lambda_{11} + \lambda_{10}$ and $p_\mathrm{e} = \lambda_{11} + \lambda_{01}$ denote the marginal update probabilities at Bob and Eve, respectively.
\end{theorem}

\begin{IEEEproof}
    See Appendix~\ref{sec:proof_thm_stationary}.
\end{IEEEproof}
The three terms correspond to the latest successful receptions occurring at the same time, more recently at Bob, and more recently at Eve, respectively.

\subsection{Metric Analysis}
Using the joint-state characterization above, we now derive the long-term \ac{lra} probability and persistence.

Define the reconstruction-advantage region as
\begin{equation*}
    \mathcal S_{\mathrm{LRA}}
    \triangleq
    \{\mathbf s=(x,a,b)\in\mathcal N^3:a=x,\ b\neq x\}.
\end{equation*}
Let $N_\mathrm{LRA}\triangleq|\mathcal S_{\mathrm{LRA}}|$, and fix an arbitrary ordering of the states in $\mathcal S_{\mathrm{LRA}}$.
Under this ordering, define the restricted stationary vector and transition submatrix as
\begin{equation*}
    \boldsymbol{\pi}_\mathrm{LRA} \triangleq
    [\pi(\mathbf s)]_{\mathbf s\in\mathcal S_\mathrm{LRA}}
    \in\mathbb R^{N_\mathrm{LRA}},
\end{equation*}
and
\begin{equation*}
    \mathbf P_\mathrm{LRA}\triangleq
    [\mathbf P_{\mathbf s,\mathbf s'}]_{\mathbf s,\mathbf s'\in\mathcal S_\mathrm{LRA}}
    \in\mathbb R^{N_\mathrm{LRA}\times N_\mathrm{LRA}}.
\end{equation*}

\begin{theorem}[Stationary \ac{lra} Metrics]
\label{thm:stationary_metrics}
The long-term average \ac{lra} probability is
\begin{equation}
    \bar J_\mathrm{LRA}
    =
    \mathbf 1_{N_\mathrm{LRA}}^{\mathsf T}\boldsymbol{\pi}_\mathrm{LRA}.
    \label{eq:stationary_metric_J_matrix}
\end{equation}
The long-term average \ac{lra} persistence is
\begin{equation}
    \bar D_\mathrm{LRA}
    =
    \mathbf 1_{N_\mathrm{LRA}}^{\mathsf T}
    \left(\mathbf I_{N_\mathrm{LRA}}-\mathbf P_\mathrm{LRA}^{\mathsf T}\right)^{-1}
    \boldsymbol{\pi}_\mathrm{LRA}.
    \label{eq:stationary_metric_D_matrix}
\end{equation}
\end{theorem}
\begin{IEEEproof}
    See Appendix~\ref{sec:proof_stationary_metrics}.
\end{IEEEproof}

Theorem~\ref{thm:stationary_metrics} reduces the evaluation of the \ac{lra} metrics to the restricted stationary vector and transition submatrix. We next characterize how these quantities, and hence the \ac{lra} metrics, depend on the transmission and channel-success probabilities.

\begin{corollary}[Rational Dependence on System Parameters]
\label{cor:rational_structure}
The long-term average \ac{lra} probability $\bar J_{\mathrm{LRA}}$ and persistence $\bar D_{\mathrm{LRA}}$ are rational functions of the transmission probability $p_{\mathrm a}$ and the channel-success probabilities $p_{\mathrm h}^{\mathrm b}$ and $p_{\mathrm h}^{\mathrm e}$.
\end{corollary}

\begin{IEEEproof}
The update probabilities $\lambda_{ij}$ and the entries of the restricted transition matrix $\mathbf P_\mathrm{LRA}$ are polynomial functions of $p_{\mathrm a}$, $p_{\mathrm h}^{\mathrm b}$, and $p_{\mathrm h}^{\mathrm e}$.
The entries of the inverse matrices appearing in Theorems~\ref{thm:stationary} and \ref{thm:stationary_metrics} are therefore rational functions of these parameters.
Hence, $\bar J_{\mathrm{LRA}}$ and $\bar D_{\mathrm{LRA}}$ are rational functions of the system parameters.
\end{IEEEproof}
Together, Theorem~\ref{thm:stationary_metrics} and Corollary~\ref{cor:rational_structure} provide computable \ac{lra} metrics for general finite-state Markov sources and establish a basis for parameter optimization.

\section{Structured-Source Specializations}\label{sec:sepcialization}
In this section, we specialize the general matrix-form results in Theorems~\ref{thm:stationary} and \ref{thm:stationary_metrics} to two-state and uniform-transition Markov sources, obtaining explicit algebraic expressions for the \ac{lra} metrics in Corollaries~\ref{cor:lra_probability_TS}-\ref{cor:lra_persistence_ut}. These expressions clarify the effects of source dynamics and transmission decisions and enable the rare-transmission analysis in Theorem~\ref{thm:rare_transmission_limit}.

\subsection{Algebraic Metric Expressions}
\subsubsection{Two-State Markov Source}
First, we consider a two-state Markov source with transition matrix
\begin{equation*}
    \mathbf Q^\mathrm{TS}=\begin{bmatrix}
        1-p &p\\
        q   &1-q
    \end{bmatrix}.
\end{equation*}
Its state space is $\mathcal N^{\mathrm{TS}}=\{1,2\}$ and stationary distribution is $\mathbf v^\mathrm{TS} =[\frac{q}{p+q},\frac{p}{p+q}]^{\mathsf{T}}$.

\begin{corollary}\label{cor:lra_probability_TS}
For the two-state source $\mathbf Q^\mathrm{TS}$, the long-term average \ac{lra} probability is
\begin{equation*}
\bar{J}^\mathrm{TS}_\mathrm{LRA}
= \frac{\beta_1 p_\mathrm{a}+\beta_2} {\beta_3 p_\mathrm{a}^{2}+\beta_4 p_\mathrm{a}+\beta_5},\label{eq:opt_p}
\end{equation*}
where
\begin{align*}
\beta_1
&=
pq\left[
    \bigl(p_{\mathrm h}^{\mathrm b}\bigr)^2
    \bigl(1-p_{\mathrm h}^{\mathrm e}\bigr)(p+q-2)
    +
    \bigl(p_{\mathrm h}^{\mathrm e}\bigr)^2
    \bigl(1-p_{\mathrm h}^{\mathrm b}\bigr)(p+q)
\right],
\\
\beta_2
&=
pq(p+q)
\left(
    2p_{\mathrm h}^{\mathrm b}p_{\mathrm h}^{\mathrm e}
    -p_{\mathrm h}^{\mathrm b}
    -p_{\mathrm h}^{\mathrm e}
\right),
\\
\beta_3
&=
\left(
    p_{\mathrm h}^{\mathrm b}p_{\mathrm h}^{\mathrm e}
    -p_{\mathrm h}^{\mathrm b}
    -p_{\mathrm h}^{\mathrm e}
\right)
p_{\mathrm h}^{\mathrm b}p_{\mathrm h}^{\mathrm e}
(p+q-1)^2(p+q),
\\
\beta_4
&=
\left(
    p_{\mathrm h}^{\mathrm b}p_{\mathrm h}^{\mathrm e}
    -p_{\mathrm h}^{\mathrm b}
    -p_{\mathrm h}^{\mathrm e}
\right)
\left(
    p_{\mathrm h}^{\mathrm b}
    +p_{\mathrm h}^{\mathrm e}
\right)
(1-p-q)(p+q)^2,
\\
\beta_5
&=
\left(
    p_{\mathrm h}^{\mathrm b}p_{\mathrm h}^{\mathrm e}
    -p_{\mathrm h}^{\mathrm b}
    -p_{\mathrm h}^{\mathrm e}
\right)
(p+q)^3 .
\end{align*}
\end{corollary}
\begin{IEEEproof}
    Substituting \eqref{eq:stationary_distribution} with $\mathbf Q=\mathbf Q^\mathrm{TS}$ into \eqref{eq:stationary_metric_J_matrix} gives the result.
\end{IEEEproof}

\begin{corollary} \label{cor:lra_persistence_TS}
For the two-state source $\mathbf Q^\mathrm{TS}$, the long-term average \ac{lra} persistence is
\begin{equation}
\begin{aligned}
    \bar D^{\mathrm{TS}}_\mathrm{LRA}
    =
    &\frac{1}{2\left(\beta_3p_{\mathrm a}^2+\beta_4p_{\mathrm a}+\beta_5\right)}\\
    &\times\left[
    \frac{\beta_1p_{\mathrm a}+\beta_2+\beta_7}
    {p+p_{\mathrm a}p_{\mathrm h}^{\mathrm e}(1-p)}
    +
    \frac{\beta_1p_{\mathrm a}+\beta_2-\beta_7}
    {q+p_{\mathrm a}p_{\mathrm h}^{\mathrm e}(1-q)}
    \right],
\end{aligned}
    \label{eq:DLRA_TS_closed_form}
\end{equation}
where $\beta_7$ is defined in \eqref{eq:def_beta7}.

\begin{figure*}
\begin{equation}
\beta_7
=
\frac{
pq(p-q)
\left[
1-p_{\mathrm a}
\left(
p_{\mathrm h}^{\mathrm b}
+p_{\mathrm h}^{\mathrm e}
-p_{\mathrm h}^{\mathrm b}p_{\mathrm h}^{\mathrm e}
\right)
\right]
\left\{
-(p+q)
\left(
2p_{\mathrm h}^{\mathrm b}p_{\mathrm h}^{\mathrm e}
-p_{\mathrm h}^{\mathrm b}
-p_{\mathrm h}^{\mathrm e}
\right)
+
p_{\mathrm a}(1-p-q)
\left[
\left(p_{\mathrm h}^{\mathrm b}\right)^2(1-p_{\mathrm h}^{\mathrm e})
+
\left(p_{\mathrm h}^{\mathrm e}\right)^2(1-p_{\mathrm h}^{\mathrm b})
\right]
\right\}
}{
(p+q)
+
p_{\mathrm a}(1-p-q)
\left(
p_{\mathrm h}^{\mathrm b}
+p_{\mathrm h}^{\mathrm e}
-p_{\mathrm h}^{\mathrm b}p_{\mathrm h}^{\mathrm e}
\right)
}.
\label{eq:def_beta7}
\end{equation}
\end{figure*}
\end{corollary}

\begin{IEEEproof}
The \ac{lra} region for the two-state source $\mathbf Q^\mathrm{TS}$ is
$
    \mathcal S_{\mathrm{LRA}}^{\mathrm{TS}}
    =
    \{(1,1,2),(2,2,1)\}.
$
Under the ordering $(1,1,2)$ and $(2,2,1)$, the corresponding restricted transition matrix is
\begin{equation*}
    \mathbf P_{\mathrm{LRA}}^{\mathrm{TS}}
    =
    \begin{bmatrix}
        (1-p)(1-p_{\mathrm a}p_{\mathrm h}^{\mathrm e}) & 0\\
        0 & (1-q)(1-p_{\mathrm a}p_{\mathrm h}^{\mathrm e})
    \end{bmatrix}.
\end{equation*}
Therefore, by \eqref{eq:stationary_metric_D_matrix},
\begin{equation}
    \bar D^{\mathrm{TS}}_\mathrm{LRA}
    =
    \frac{\pi(1,1,2)}
    {p+p_{\mathrm a}p_{\mathrm h}^{\mathrm e}(1-p)}
    +
    \frac{\pi(2,2,1)}
    {q+p_{\mathrm a}p_{\mathrm h}^{\mathrm e}(1-q)} .
    \label{eq:DLRA_TS_proof}
\end{equation}
Substituting \eqref{eq:stationary_distribution} with $\mathbf Q=\mathbf Q^\mathrm{TS}$ into \eqref{eq:DLRA_TS_proof} gives the result.
\end{IEEEproof}

\subsubsection{Uniform-Transition Markov Source}
We then consider the uniform-transition Markov source with state space $\mathcal{N}^\mathrm{UT}=\{1,2,\ldots, N\}$. The source remains in its current state with probability $\rho$ and transitions to each of the other $N-1$ states with probability $\eta$. The transition matrix is given by
\begin{equation}
    [\mathbf Q^{\mathrm{UT}}]_{i,j}
    =
    \begin{cases}
    \rho, & i=j,\\
    \eta, & i\neq j,
    \end{cases}
    \qquad i,j\in\mathcal N,
    \label{eq:uniform_transition_Q}
\end{equation}
where $\rho+(N-1)\eta=1$ with $\eta,\rho\in(0,1)$. 
Owing to its state symmetry, the stationary distribution is uniform, \emph{i.e.}, $v^{\mathrm{UT}}_i=\frac{1}{N}$, for all $i\in\mathcal{N}^\mathrm{UT}$.

\begin{corollary}\label{cor:lra_probability_ut}
    For the {uniform-transition source} $\mathbf Q^{\mathrm{UT}}$, the long-term  average \ac{lra} probability is
    \begin{equation}
        \bar J_{\mathrm{LRA}}^{\mathrm UT}=
        \frac{(N-1)\eta\lambda_{10}\psi_{\mathrm d}}{1-(1-p_\mathrm{e})d}
        +
        \frac{(N-1)\eta\lambda_{01}\psi_{\mathrm o}}{1-(1-p_\mathrm{b})d},
    \label{eq:uniform_avg_cra_compact}
    \end{equation}
    where $d=\rho-\eta$ and
    \begin{equation*}
        \psi_\mathrm{d} = \frac{1+\lambda_{00}(\eta-1)}{(1-\lambda_{00}d)(1-\lambda_{00})},\,
        \psi_\mathrm{o}  = \frac{\lambda_{00}\eta}{(1-\lambda_{00}d)(1-\lambda_{00})}.
    \end{equation*}
\end{corollary}

\begin{IEEEproof}
    See Appendix~\ref{sec:proof_lra_probability_ut}.
\end{IEEEproof}

\begin{corollary}\label{cor:lra_persistence_ut}
For the {uniform-transition source} $\mathbf Q^{\mathrm{UT}}$, the long-term  average \ac{lra} persistence is
\begin{equation}
    \bar D_{\mathrm{LRA}}^{\mathrm{UT}}
    =
    \frac{
    \bar J_{\mathrm{LRA}}^{\mathrm{UT}}
    }{
    1-\rho(1-p_{\mathrm e})-(N-2)\eta\lambda_{10}
    }.
    \label{eq:ut_lra_persistence}
\end{equation}
\end{corollary}

\begin{IEEEproof}
    See Appendix~\ref{sec:proof_ut_lra_persistence}.
\end{IEEEproof}

\subsection{Rare-Transmission Limits}
This subsection characterizes the rare-transmission limits $p_{\mathrm a}\to0^+$, separating the effects of source dynamics and receiver separation on \ac{lra} probability and persistence.
To characterize this limiting structure, we first introduce three quantities.
First, define the stationary-weighted asymptotic source-mismatch probability as
\begin{align*}
    \mu_{\mathrm{mis}} 
    &\triangleq
    \mathbb{E}_{i\sim\mathbf v} \bigg\{
    \lim_{\substack{t\to\infty\\ k\to\infty}}
    \Pr\left\{X_t = i, X_{t+k}\neq i\right\}
    \bigg\}\\
    &=\sum_{i\in \mathcal{N}}[\mathbf v]^2_i\left(1-[\mathbf v]_i\right)
\end{align*}
Second, define the mismatch-weighted source sojourn time as
\begin{align*}
\tau_{\rm soj}
&\triangleq
\sum_{\ell=1}^{\infty}
\mathbb E_{i\sim \mathbf v}
\bigg\{
\lim_{\substack{t\to\infty\\ k\to\infty}}
\Pr\bigg\{
\begin{aligned}
    & X_t=\cdots=X_{t+\ell-1}=i,\;\\
    & X_{t+k}\ne i
\end{aligned}
\bigg\}
\bigg\}\\
&= \sum_{i\in \mathcal{N}}\frac{[\mathbf v]^2_i\left(1-[\mathbf v]_i\right)}{1-[\mathbf Q]_{i,i}}.
\end{align*}
Finally, define the receiver-separation probability as
\begin{equation*}
    \mu_{\mathrm{sep}} \triangleq \Pr \left\{H_t^\mathrm{b}\neq H_t^\mathrm{e}\mid H_t^\mathrm{b}+ H_t^\mathrm{e} \ge 1\right\}
    =\frac{
        p_{\mathrm h}^{\mathrm b}+p_{\mathrm h}^{\mathrm e}
        -2p_{\mathrm h}^{\mathrm b}p_{\mathrm h}^{\mathrm e}
        }{
        p_{\mathrm h}^{\mathrm b}+p_{\mathrm h}^{\mathrm e}
        -p_{\mathrm h}^{\mathrm b}p_{\mathrm h}^{\mathrm e}
        }.
\end{equation*}
Here, $\mu_{\mathrm{sep}}$ is channel-dependent, whereas $\mu_{\mathrm{mis}}$ and $\tau_{\mathrm{soj}}$ are source-dependent.
For the two-state source $\mathbf Q^\mathrm{TS}$, these two source-dependent quantities become
\begin{equation*}
    \mu_{\mathrm{mis}}^\mathrm{TS}=\frac{pq}{(p+q)^2},\ \tau_{\mathrm{soj}}^\mathrm{TS}=\frac{p^2+q^2}{(p+q)^3}.
\end{equation*}
For the uniform-transition source $\mathbf Q^\mathrm{UT}$, they become
\begin{equation*}
    \mu_{\mathrm{mis}}^\mathrm{UT}=\frac{N-1}{N^2},\ \tau_{\mathrm{soj}}^\mathrm{UT}=\frac{N-1}{N^2(1-\rho)}.
\end{equation*}
\begin{figure}
    \centering
    \includegraphics[width=1\linewidth]{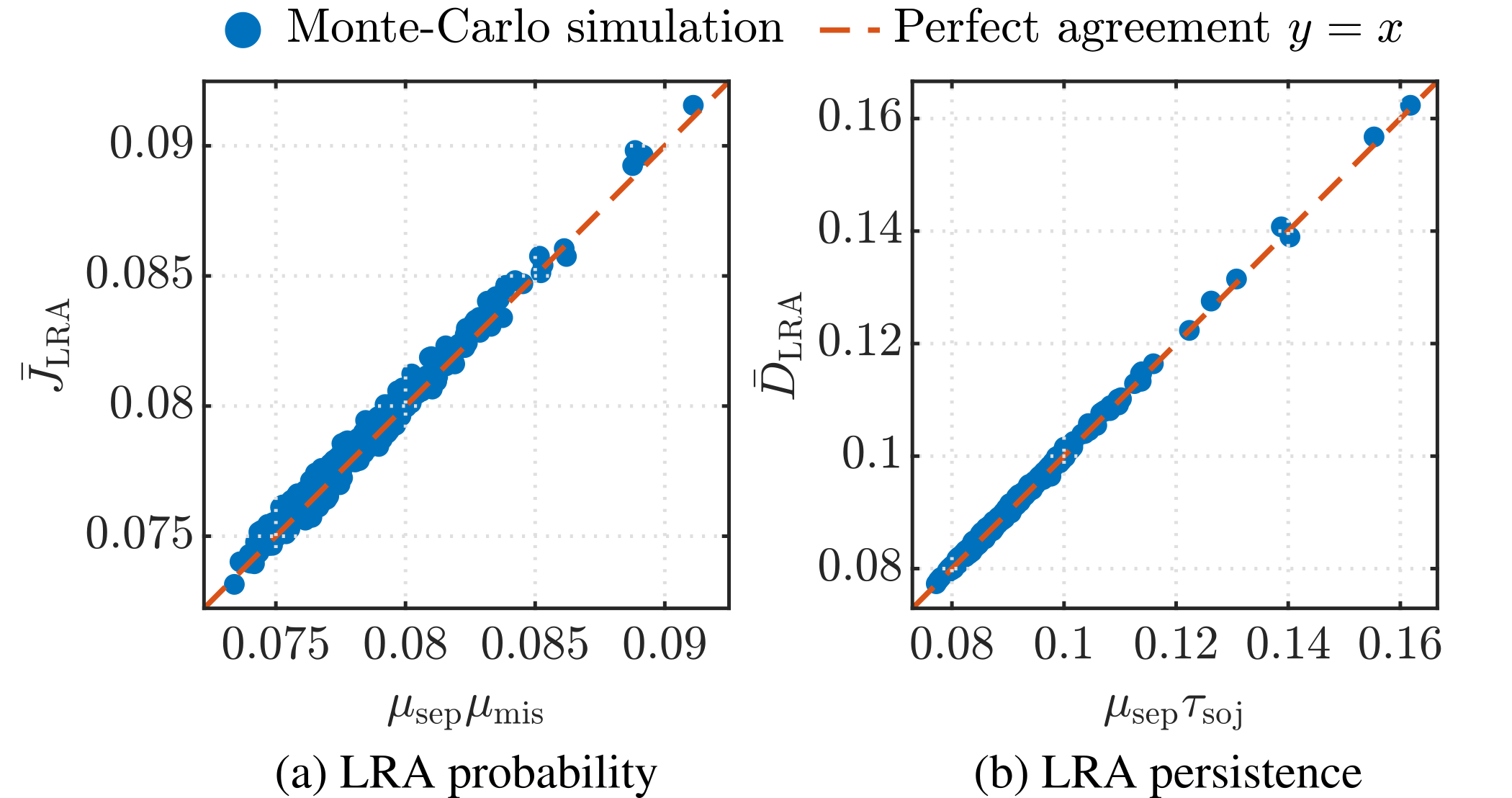}
    \caption{Validation of the rare-transmission factorization for random general Markov sources{, with $N=10$, $p_\mathrm{h}^\mathrm{b}=0.8$, $p_\mathrm{h}^\mathrm{e}=0.2$, and $p_\mathrm{a}=0.0005$}.}
    \label{fig:limit_extension}
\end{figure}

\begin{theorem}[Rare-Transmission Limit]\label{thm:rare_transmission_limit}
    For the two-state source $\mathbf Q^\mathrm{TS}$ and uniform-transition source $\mathbf Q^\mathrm{UT}$, as the transmission probability $p_{\mathrm a}\to 0^+$, the \ac{lra} probability $\bar{J}_\mathrm{LRA}$ converges to
    \begin{equation*}
    \lim_{p_{\mathrm a}\to 0^+}
    \bar J_{\mathrm{LRA}}
    =
    \mu_{\mathrm{sep}} \mu_{\mathrm{mis}},
    \end{equation*}
    and the \ac{lra} persistence $\bar{D}_\mathrm{LRA}$ converges to
    \begin{equation*}
    \lim_{p_{\mathrm a}\to 0^+}
    \bar D_{\mathrm{LRA}}
    =
    \mu_{\mathrm{sep}} \tau_{\mathrm{soj}}.
    \end{equation*}
\end{theorem}
\begin{IEEEproof}
    The result follows directly by taking the rare-transmission limit $p_{\mathrm a}\to 0^+$ in the algebraic expressions in Corollaries \ref{cor:lra_probability_TS}-\ref{cor:lra_persistence_ut} and identifying the resulting factors as $\mu_{\mathrm{mis}}$, $\mu_{\mathrm{sep}}$, and $\tau_{\mathrm{soj}}$.
\end{IEEEproof}

Theorem~\ref{thm:rare_transmission_limit} separates the roles of source dynamics and channel-induced receiver separation. Both limits share the receiver-separation factor $\mu_{\mathrm{sep}}$, while their source-dependent factors capture state mismatch and state sojourns for \ac{lra} probability and persistence, respectively. Thus, achieving and sustaining joint accuracy and security depend on different aspects of source dynamics. The numerical experiments in Fig.~\ref{fig:limit_extension} further suggest that this factorized structure also holds for general Markov sources beyond these two structured models.

Moreover, these limits reveal contrasting effects of source-transition asymmetry on \ac{lra} probability and persistence, as detailed in the following remark.

\begin{remark}[Transition Asymmetry]
For the two-state source $\mathbf Q^{\rm TS}$, the rare-transmission limits show that $\bar J_{\rm LRA}^{\rm TS}$ and $\bar D_{\rm LRA}^{\rm TS}$ can favor different source-transition structures. Let $s=p+q$ and $\delta=|p-q|$. For fixed $s$, the source-dependent factor in the limiting \ac{lra} probability is
$\mu_{\rm mis}^{\rm TS} = \frac{pq}{(p+q)^2} = \frac{s^2-\delta^2}{4s^2},$
which is maximized at $\delta=0$, \emph{i.e.}, $p=q$. In contrast, the source-dependent factor in the limiting \ac{lra} persistence is
$ \tau_{\rm soj}^{\rm TS} = \frac{p^2+q^2}{(p+q)^3} = \frac{s^2+\delta^2}{2s^3},$
which increases with $\delta$. Hence, \ac{lra} probability favors balanced transitions, whereas \ac{lra} persistence favors asymmetric transitions.
\end{remark}

\section{{Transmission Probability} Optimization}
\label{sec:optimization_pa}
{Building on the performance characterizations in Sections~\ref{sec:general_performance} and \ref{sec:sepcialization}, we optimize the transmission probability of the \ac{rs} policy to maximize either \ac{lra} probability or persistence. The general-source expressions enable numerical optimization, while the structured-source expressions support polynomial root finding and, for \ac{lra} probability under the two-state source, closed-form solutions.}
\subsection{Optimality Analysis}
For a generic objective $F(p_{\mathrm a})$, representing {either} $\bar{J}_{\mathrm{LRA}}$ {or} $\bar{D}_{\mathrm{LRA}}$ {for a fixed source and fixed channel parameters}, we consider
\begin{equation*}
    \mathscr{P}1:\quad
    \sup_{0< p_{\mathrm a}\leq 1}
    \ F(p_{\mathrm a}).
\end{equation*}
By Corollary~\ref{cor:rational_structure}, $F(p_{\mathrm a})$ is a rational function for general finite-state Markov sources and can be written as $ F(p_{\mathrm a}) = {\mathcal N_F(p_{\mathrm a})}/{\mathcal D_F(p_{\mathrm a})}$, where both $\mathcal N_F(p_{\mathrm a})$ and $\mathcal D_F(p_{\mathrm a})$ are polynomials in $p_{\mathrm a}$. Theorems~\ref{thm:stationary} and \ref{thm:stationary_metrics} enable direct evaluation of $F(p_{\mathrm a})$, supporting numerical optimization over the transmission probability. For two-state and uniform-transition sources, Corollaries \ref{cor:lra_probability_TS}-\ref{cor:lra_persistence_ut} further provide explicit expressions for these polynomials.

{To account for the open boundary at zero, we continuously extend the objective by defining
$
F(0)\triangleq\lim_{p_{\mathrm a}\to0^+}F(p_{\mathrm a}).
$
This value represents the rare-transmission limit, rather than the performance of a zero-transmission policy. After cancellation of removable factors, the rational representation can be chosen with $D_F(p_{\mathrm a})\ne0$ on $[0,1]$. Consequently,
\[
\sup_{0<p_{\mathrm a}\le1}F(p_{\mathrm a})
=
\max_{0\le p_{\mathrm a}\le1}F(p_{\mathrm a}).
\]
Let $p_\mathrm{a}^*$ and $p_\mathrm{a}^\dagger$ denote optimal solutions to the left and right problems, respectively. If $p_{\mathrm a}^{\dagger}>0$, it also solves $\mathcal P1$. If $p_{\mathrm a}^{\dagger}=0$, the supremum is approached as $p_{\mathrm a}\to0^+$.}

\begin{theorem}[Optimality]
\label{thm:polynomial_optimal_pa}
For any objective $F(p_{\mathrm a})$ {in $\mathscr P 1$, a maximizer of its continuous extension to} $p_{\mathrm a}\in[0,1]$ is 
\begin{equation}
    p_{\mathrm a}^{\dagger}
    \in
    \arg\max_{p_{\mathrm a}\in\mathcal C_F} F(p_{\mathrm a}),
\end{equation}
where
$
    \mathcal C_F
    \triangleq
    \{0,1\}
    \cup
    \left\{
    p_{\mathrm a}\in(0,1):
    \mathcal P_F(p_{\mathrm a})=0
    \right\},
$
and
$
    \mathcal P_F(p_{\mathrm a})
    \triangleq
    \mathcal N_F'(p_{\mathrm a})\mathcal D_F(p_{\mathrm a})
    -
    \mathcal N_F(p_{\mathrm a})\mathcal D_F'(p_{\mathrm a}).
$

\end{theorem}

\begin{IEEEproof}
Since $F(p_{\mathrm a})$ is a rational function with $\mathcal D_F(p_{\mathrm a})\neq 0$, any interior optimum must satisfy $F'(p_{\mathrm a})=0$, which is equivalent to $\mathcal P_F(p_{\mathrm a})=0$; together with the boundary points $p_{\mathrm a}=0$ and $p_{\mathrm a}=1$, this gives the candidate set $\mathcal C_F$.
\end{IEEEproof}
\begin{table}[t]
\centering
\caption{Computational Complexity Comparison.}
\label{tab:complexity_pa_optimization}
\scriptsize
\begin{tabular}{
p{0.31\columnwidth}
p{0.24\columnwidth}
p{0.26\columnwidth}}
\toprule
\textbf{Method} & \textbf{Applicability} & \textbf{Complexity} \\
\midrule
Theorem~\ref{thm:stationary} + Theorem~\ref{thm:stationary_metrics} + grid search &
General sources &
$\mathcal O(N^3/\epsilon)$ for $\bar J_{\mathrm{LRA}}$;
$\mathcal O(N^6/\epsilon)$ for $\bar D_{\mathrm{LRA}}$ \\
\midrule
Theorem~\ref{thm:polynomial_optimal_pa} &
Uniform-transition and two-state sources &
$\widetilde{\mathcal O}(\log(1/\epsilon))$ \\
\midrule
Corollary~\ref{cor:opt_P1} &
$\bar J_{\mathrm{LRA}}^{\mathrm{TS}}$ &
$\mathcal O(1)$ \\
\bottomrule
\end{tabular}
\end{table}

\begin{figure*}
    \centering
    \includegraphics[width=1\linewidth]{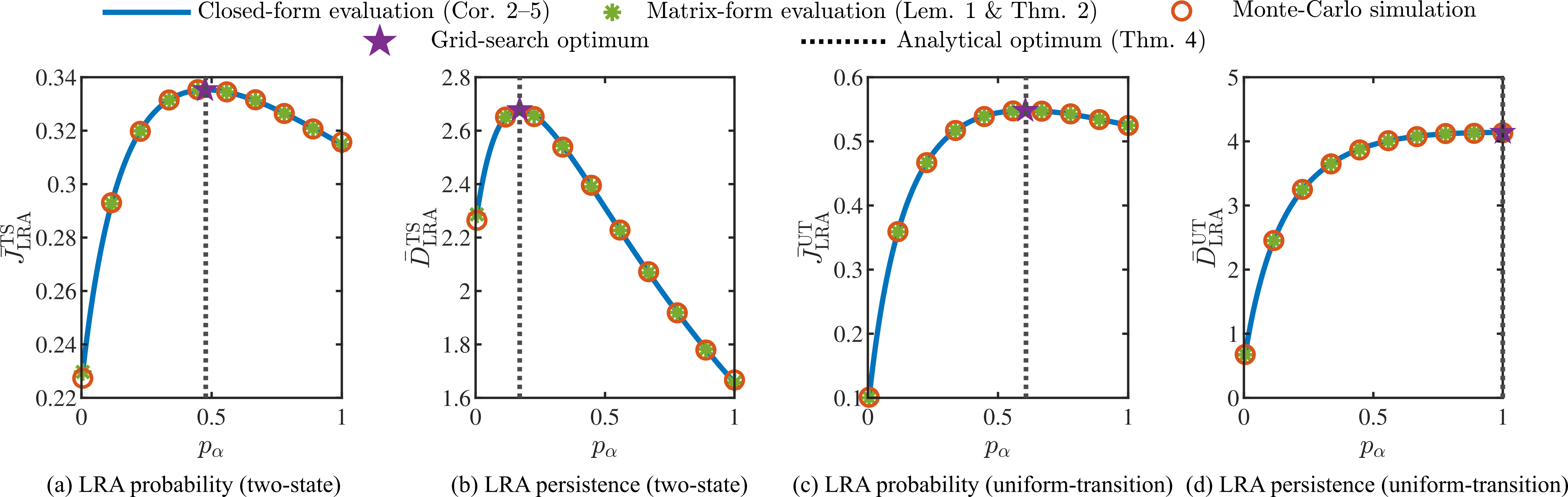}
    \caption{\ac{lra} probability and persistence metrics versus the {transmission} probability $p_{a}$ for {two-state and uniform-transition sources}. Subfigures (a) and (b) show the \ac{lra} probability and \ac{lra} persistence for the {two-state} source with $p=0.1$, $q=0.1$, $p_{\mathrm h}^{\mathrm b}=0.9$, and $p_{\mathrm h}^{\mathrm e}=0.1$. 
    Subfigures (c) and (d) show the corresponding results for the {uniform-transition} source with $N=10$, $\rho=0.85$, $\eta=0.0167$, $p_{\mathrm h}^{\mathrm b}=0.9$, and $p_{\mathrm h}^{\mathrm e}=0.1$. }
    \label{fig:Monto_carlo}
\end{figure*}

Theorem~\ref{thm:polynomial_optimal_pa} provides a general characterization of the optimal {transmission} probability: it can be found by solving a scalar polynomial equation and comparing the resulting stationary candidates with the boundary points.
For structured source models, this characterization may be further simplified by exploiting the specific transition structure of the underlying Markov chain.
Such simplifications are source-dependent and therefore need to be derived on a case-by-case basis.
As a representative example, the {two-state source} admits the following closed-form characterization.

\begin{corollary} \label{cor:opt_P1}
For the {two-state source} $\mathbf{Q}^\mathrm{TS}$, {a maximizer of its continuous extension to $p_{\mathrm a}\in[0,1]$} for $F(p_\mathrm{a}) = \bar{J}_{\mathrm{LRA}}^\mathrm{TS}$ is 
\begin{equation*}
p_\mathrm a^\dagger= \min\left\{1,\max\left\{\frac{-\beta_2\beta_3+\sqrt{\Delta}}{\beta_1\beta_3},0\right\}\right\},
\end{equation*}
where $\Delta=\beta_2^2\beta_3^2+\beta_1\beta_3(\beta_1\beta_5-\beta_2\beta_4)\ge 0$ {when $\beta_1\beta_3\neq0$. Otherwise, $p_\mathrm a^\dagger=\mathbbm{1}\{F(1)>F(0)\}$}.
\end{corollary}
\begin{IEEEproof}
    See Appendix~\ref{sec:proof_cor_opt_P1}.
\end{IEEEproof}

\begin{remark}\label{rem:extrem_analysis_opt}
    In the symmetric-channel case with $p_\mathrm{h}^\mathrm{b}=p_\mathrm{h}^\mathrm{e}\in(0,1)$, $\bar{J}_{\mathrm{LRA}}^\mathrm{TS}(p_\mathrm{a})$ is monotonic over $p_\mathrm{a}\in[0,1]$: it is increasing when $p+q>1$ and decreasing when $p+q<1$, and constant when $p+q=1$. Accordingly, $p_\mathrm{a}^\dagger=1$ if $p+q>1$, $p_\mathrm{a}^\dagger=0$ if $p+q<1$, and $p_\mathrm{a}^\dagger$ can be any value in $[0,1]$ if $p+q=1$.
 
    Thus, in the symmetric-channel case, \emph{the optimal transmission policy is determined entirely by the source temporal correlation}: persistent sources favor {sparse transmission}, alternating sources favor full transmission, and the \ac{iid} case is indifferent to the transmission probability.
\end{remark}

\subsection{Complexity Analysis}
Let $\epsilon$ denote {the grid resolution or the tolerance used to approximate candidate roots}.
A classical approach is to construct the induced Markov chain, compute its stationary distribution for each ${p}_\mathrm{a}$, and perform grid search over ${p}_\mathrm{a}$. Since the induced chain has dimension $N^3$, this approach has complexity $\mathcal O(N^9/\epsilon)$ under direct dense-matrix computation.

In contrast, the proposed analysis leads to lower-complexity optimization methods, as illustrated in Table~\ref{tab:complexity_pa_optimization}.
For general Markov sources, Theorem~\ref{thm:stationary} provides an explicit stationary distribution characterization without solving the $N^3$-dimensional stationary equation, reducing the complexity to $\mathcal O(N^3/\epsilon)$ for $\bar J_\mathrm{LRA}$.
For $\bar D_\mathrm{LRA}$, the restricted \ac{lra} transition matrix has dimension $N_\mathrm{LRA} = N(N-1)$, yielding complexity $\mathcal O(N^6/\epsilon)$.

For the structured {two-state and uniform-transition sources}, Theorem~\ref{thm:polynomial_optimal_pa} further avoids exhaustive grid search by converting the optimization into scalar polynomial root finding, yielding $\widetilde{\mathcal O}(\log(1/\epsilon))$.
Finally, for the {two-state source} with objective $\bar{J}_\mathrm{LRA}^\mathrm{TS}$, Corollary~\ref{cor:opt_P1} gives a closed-form characterization of the {extended-objective maximizer, requiring $\mathcal O(1)$ arithmetic operations.}

\section{Numerical and Simulation Results}
In this section, we first validate the analytical results using Monte Carlo simulations with horizon $T=50,000$, averaged over $1,000$ independent realizations. We then show how the proposed joint-event metric differs from conventional separate-marginal formulations, and reveal structural behaviors {that marginal objectives alone do not capture}. Finally, we demonstrate a system-level implication of the framework through a geofencing example, where {prescribed \ac{lra} requirements are translated into spatial constraints on eavesdropper locations}.
\begin{figure*}[t]
\begin{centering}
\includegraphics[width=1\textwidth]{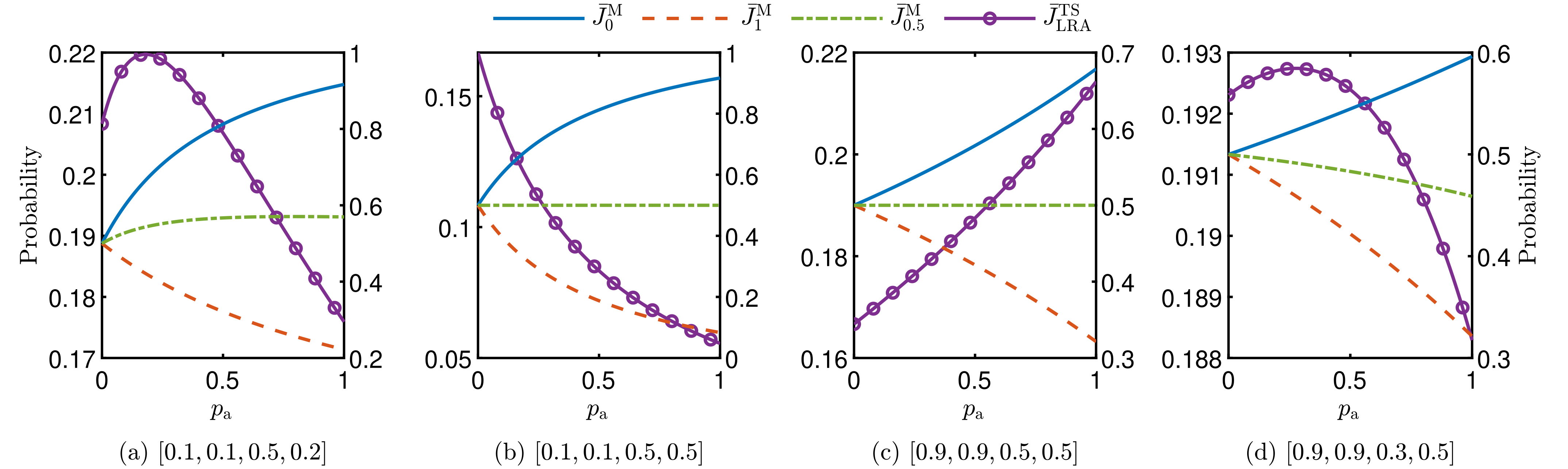}
\par\end{centering}
\caption{\label{fig:P_over_palpha}{\ac{lra} probability} $\bar{J}_\mathrm{LRA}^\mathrm{TS}$, {marginal eavesdropper error probability} $\bar{J}_1^\mathrm{M}$, {marginal legitimate accuracy} $\bar{J}_0^\mathrm{M}$, and the balance metric $\bar{J}_{0.5}^\mathrm{M}$ versus the transmission probability under different settings. We plot the metric $\bar{J}_\mathrm{LRA}^\mathrm{TS}$ on the left y-axis, while the other metrics share the right y-axis. The title of each subfigure indicates the corresponding parameter setting $[p, q, p_\mathrm{h}^\mathrm{b}, p_\mathrm{h}^\mathrm{e}]$.}
\end{figure*}

\begin{figure*}[t]
\begin{centering}
\includegraphics[width=1\textwidth]{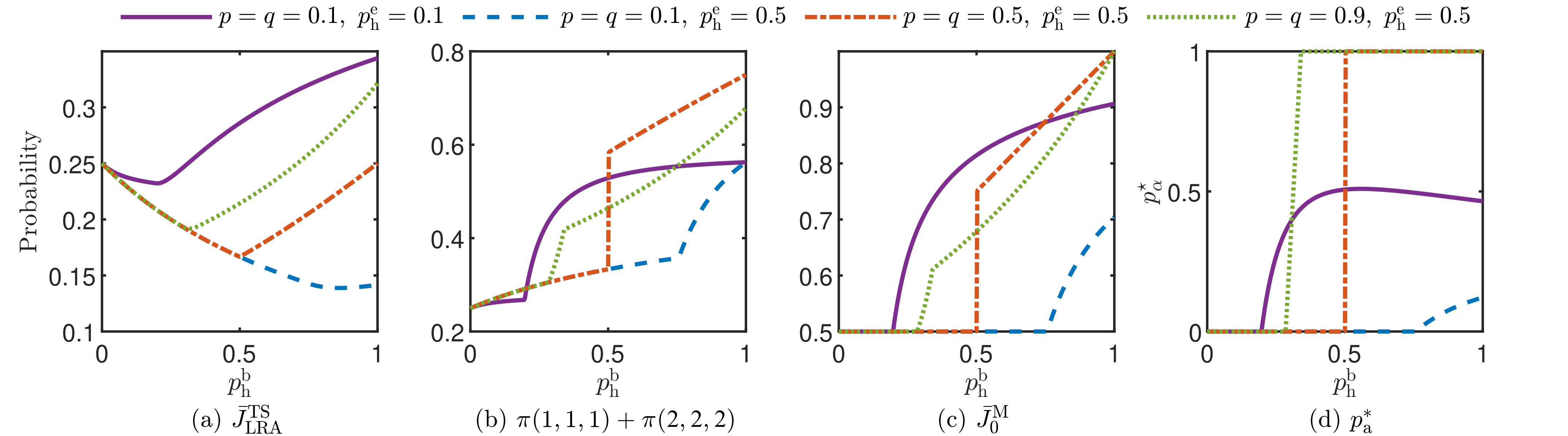}
\par\end{centering}
\caption{\label{fig:P_over_ps}{\ac{lra} probability} $\bar{J}_\mathrm{LRA}^\mathrm{TS}=\pi(1,1,2)+\pi(2,2,1)$ in (a), {the probability of simultaneous correct reconstruction} $\pi(1,1,1)+\pi(2,2,2)$ in (b), and marginal accuracy $\bar{J}_0^\mathrm{M}=\pi(1,1,2)+\pi(2,2,1)+\pi(1,1,1)+\pi(2,2,2)$ in (c), evaluated under the optimal transmission probability $p_\mathrm{a}^*$ shown in (d), versus the Bob-side channel success probability $p_\mathrm{h}^\mathrm{b}$.}
\end{figure*}

\subsubsection{Analytical Validation}
Fig.~\ref{fig:Monto_carlo} shows the \ac{lra} probability and \ac{lra} persistence versus the {transmission} probability $p_\mathrm{a}$ for both {two-state and uniform-transition} sources, by three evaluation methods. First, closed-form evaluation directly computes the four metrics using Corollaries~\ref{cor:lra_probability_TS}-\ref{cor:lra_persistence_ut}, corresponding to Figs.~\ref{fig:Monto_carlo}(a)-\ref{fig:Monto_carlo}(d). 
Second, matrix-form evaluation {numerically solves the stationary equations for the transition matrix} in Lemma~\ref{lem:mc_3d}, and then computes the metrics using Theorem~\ref{thm:stationary_metrics}. 
Third, Monte Carlo simulation repeatedly simulates the system dynamics under the same parameter settings and computes empirical averages according to~\eqref{eq:avg_reconstruction_advantage} and~\eqref{eq:avg_reconstruction_advantage_duration}.

The three methods are in {close} agreement across all four subfigures. This validates the proposed system-dynamics characterization, the matrix-form metric evaluation, and the closed-form expressions for the {structured source models}.
Moreover, the grid-search optima coincide with the analytical optima, {consistent with Theorem~\ref{thm:polynomial_optimal_pa} and the resulting polynomial-root-based optimization procedure.}

Figure~\ref{fig:Monto_carlo} further shows that different reconstruction-advantage metrics lead to different optimal {transmission} probabilities. For the {two-state source}, the {transmission} probabilities that maximize the \ac{lra} probability and persistence are approximately $0.477$ and $0.17$, respectively. Using the probability-optimal {transmission} probability when persistence is the actual objective not only increases the transmission rate but also reduces persistence from $2.68$ to $2.35$, an $12.4\%$ loss. The same effect holds for the {uniform-transition source}: using persistence-optimal $p_\mathrm a=1$ instead of the probability-optimal $p_\mathrm a=0.612$ consumes more resources while reducing the \ac{lra} probability from $0.549$ to $0.526$, an approximately $4.10\%$ loss. Therefore, the transmission probability should be selected based on {whether the objective is to maximize the occurrence or persistence of joint accuracy and security}.

For brevity, the remaining results focus on \ac{lra} probability for the two-state sources.

\begin{figure*}
\begin{centering}
\includegraphics[width=1\textwidth]{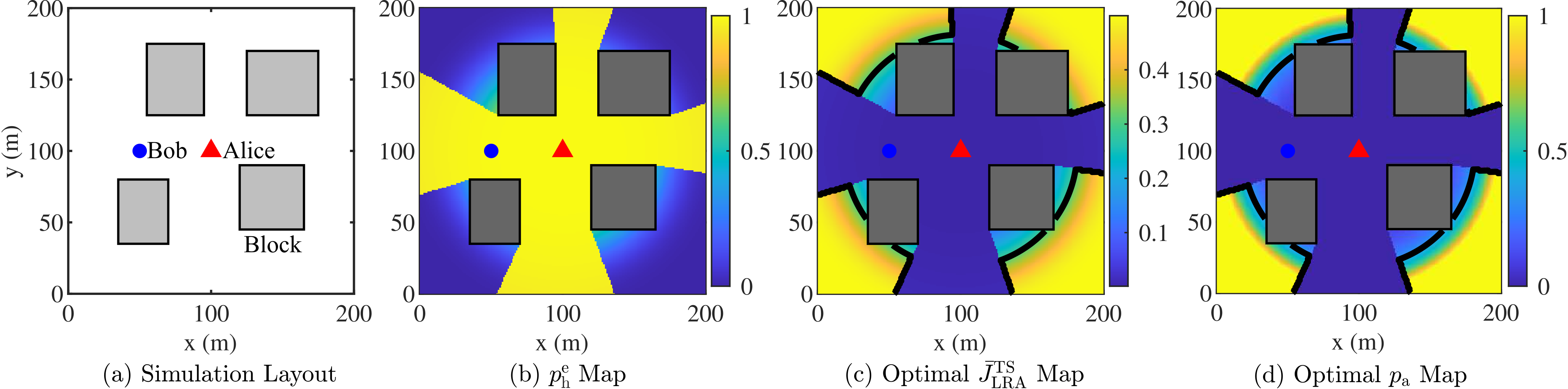}
\par\end{centering}
\caption{\label{fig:geofencing} Simulation layout and spatial maps {for candidate eavesdropper locations}, for the {two-state} source with $p=0.1$ and $q=0.1$. (a) shows the simulation layout; (b) shows Eve's channel success-probability map; (c) shows the corresponding optimal $\bar{J}_{\text{LRA}}^\mathrm{TS}$ map; and (d) shows the corresponding optimal transmission-probability $p_\mathrm a$ map. The black contour in (c) and (d) represents the geofencing boundary induced by the threshold $\tau=0.3$.}
\end{figure*}
\subsubsection{Comparison with Marginal Analysis}
For comparison, we consider a marginal-analysis baseline, which {represents accuracy and security through a weighted sum of legitimate reconstruction success and eavesdropper reconstruction failure probabilities}. Define the {legitimate-receiver accuracy event} and the 
{eavesdropper error event} as
$
J_{\mathrm A}^{\mathrm M}[t]
\triangleq
\mathbbm 1\{\hat X_t^{\mathrm b}=X_t\}$, $\quad J_{\mathrm S}^{\mathrm M}[t]
\triangleq
\mathbbm 1\{\hat X_t^{\mathrm e}\neq X_t\}
$.
Their corresponding long-term average probabilities are defined similarly to \eqref{eq:avg_reconstruction_advantage}, denoted by
$\bar J_{\mathrm A}^{\mathrm M}$ and
$\bar J_{\mathrm S}^{\mathrm M}$, and can be obtained by Theorem~\ref{thm:stationary_metrics}, or equivalently, from the marginal reconstruction analysis in~\cite{SalKouPap:J24}. The marginal-analysis baselines are then defined as
$
\bar J_{\omega}^{\mathrm M}
=
(1-\omega)\bar J_{\mathrm A}^{\mathrm M}
+
\omega \bar J_{\mathrm S}^{\mathrm M},
$
where $\omega\in[0,1]$. In particular, $\bar{J}_0^{\mathrm M}$ reduces to the pure accuracy metric, while $\bar{J}_1^{\mathrm M}$ reduces to the pure security metric.

Fig.~\ref{fig:P_over_palpha} shows that the optimal transmission probability for maximizing $\bar{J}_\mathrm{LRA}^\mathrm{TS}$ is environment-dependent, unlike the marginal metrics $\bar{J}_0^{\mathrm M}$ and $\bar{J}_1^{\mathrm M}$, whose optima are always attained at $p_{\mathrm a}=1$ and $p_\mathrm a\to0$, respectively. In Fig.~\ref{fig:P_over_palpha}a, although the legitimate channel is better than the eavesdropping channel, transmitting more frequently is not always beneficial: a larger $p_{\mathrm a}$ improves legitimate accuracy but also improves Eve’s reconstruction, and excessive transmissions eventually reduce the probability of being simultaneously secure and accurate. In Fig.~\ref{fig:P_over_palpha}d, by contrast, even when the eavesdropping channel is better than the legitimate channel, increasing $p_\mathrm a$ to a proper level can still improve $\bar{J}_\mathrm{LRA}^\mathrm{TS}$, {despite the accompanying improvement in Eve's marginal reconstruction accuracy}.

Fig.~\ref{fig:P_over_palpha} also shows that the secure remote reconstruction cannot be reduced to a weighted combination of two separate marginal metrics. In Fig.~\ref{fig:P_over_palpha}a, {the balanced metric $\bar J_{0.5}^{\mathrm M}$ selects $p_{\mathrm a}\approx0.8$, under which the actual \ac{lra} probability is approximately $0.189$. Directly optimizing \ac{lra} instead yields approximately $0.22$ at $p_{\mathrm a}\approx0.18$. Thus, the marginally optimized policy transmits more frequently while achieving an approximately $14\%$ lower \ac{lra} probability.
Figs.~\ref{fig:P_over_palpha}(b) and~\ref{fig:P_over_palpha}(c) further illustrate the symmetric-channel case with $p_{\mathrm h}^{b}=p_{\mathrm h}^{e}$, where $\bar J_{0.5}^{\mathrm M}$ is insensitive to $p_{\mathrm a}$. The two receivers have equal marginal reconstruction accuracy, so the equally weighted sum of Bob’s success probability and Eve’s error probability remains $1/2$. Nevertheless, the preferred transmission probability for $\bar J_{\mathrm{LRA}}^{\mathrm{TS}}$ varies significantly with the source dynamics, consistent with Remark~\ref{rem:extrem_analysis_opt}: persistent sources favor increasingly sparse transmission, whereas alternating sources favor full transmission. The weighted marginal objective therefore fails to distinguish policies whose \ac{lra} probabilities differ.}

\subsubsection{Effect of Legitimate-Channel Quality}

Fig.~\ref{fig:P_over_ps} studies the effect of Bob’s channel success probability $p_\mathrm{h}^{\mathrm b}$ on \ac{lra} probability under the optimal transmission probability $p_\mathrm{a}^*$.
As shown in Fig.~\ref{fig:P_over_ps}a, improving the legitimate channel does not necessarily increase the \ac{lra} probability monotonically.
A larger $p_\mathrm{h}^{\mathrm b}$ improves Bob’s marginal reconstruction accuracy, as shown in Fig.~\ref{fig:P_over_ps}c. However, part of this accuracy gain may occur in states where Eve is also accurate, thereby increasing the {probability of simultaneous correct reconstruction} in Fig.~\ref{fig:P_over_ps}b.
Hence, the gain in Bob’s accuracy can initially be offset by a larger non-security component, causing the \ac{lra} probability to first decrease and then increase.

Second, whether improving $p_\mathrm{h}^{\mathrm b}$ yields a meaningful secure reconstruction gain depends on Eve’s channel quality and the source dynamics.
When Eve’s channel is strong, \emph{e.g.}, $p=q=0.1$ and $p_\mathrm{h}^{\mathrm e}=0.5$, increasing $p_\mathrm{h}^{\mathrm b}$ may primarily synchronize the two reconstructions rather than enlarge their separation.
In such cases, the \ac{lra} probability improves more effectively by weakening Eve’s channel, \emph{e.g.}, setting $p_\mathrm{h}^{\mathrm e}=0.1$, or by accelerating the source evolution, \emph{e.g.}, setting $p=q=0.9$.
The former reduces Eve’s access to fresh updates, whereas the latter makes outdated reconstructions less reliable.

{Third, improving legitimate-channel reliability does not necessarily call for more frequent transmission. As shown in Fig.~\ref{fig:P_over_ps}(d), for $p=q=0.1$ and $p_{\mathrm h}^{e}=0.1$, the \ac{lra}-optimal transmission probability first increases and then decreases with $p_{\mathrm h}^{b}$. A more reliable Bob-side channel can support accurate reconstruction with fewer transmissions, reducing the need for additional updates that may also benefit Eve. The preferred policy also depends on source dynamics and Eve’s channel: with $p_{\mathrm h}^{e}=0.5$, the persistent source ($p=q=0.1$) favors sparse transmission over a broad range of $p_{\mathrm h}^{b}$, whereas the frequently alternating source ($p=q=0.9$) favors full transmission at much lower $p_{\mathrm h}^{b}$. Thus, transmission frequency should be selected jointly according to source dynamics and both receiver channels to maximize \ac{lra} probability.}

\subsubsection{\ac{lra}-Induced Geofencing}
To further evaluate the proposed framework in a representative practical scenario, we consider the urban layout shown in Fig.~\ref{fig:geofencing}a, where the transmitter (Alice) is located at the center and the legitimate receiver (Bob) is placed at a fixed location, while the eavesdropper (Eve) is assumed to be able to appear anywhere in the area. The large-scale propagation is modeled by the 3GPP Urban Micro (UMi) path-loss model in~\cite{TR36814}, where the \ac{los} and \ac{nlos} links are given by
\[
\mathrm{PL}_j=
\begin{cases}
28.0+22.0\log_{10}(d_j)+20\log_{10}(f_c), & \text{LOS},\\
22.7+36.7\log_{10}(d_j)+26\log_{10}(f_c), & \text{NLOS},
\end{cases}
\]
and the small-scale fading is modeled as Rayleigh fading. Based on this channel model, the channel success probability is defined as the probability that the instantaneous channel quality exceeds a prescribed threshold, {\em i.e.}, $\Pr(\gamma>\gamma_{\mathrm{th}})$, which yields the spatial success-probability map shown in Fig.~\ref{fig:geofencing}b. Then, by applying Corollary~\ref{cor:opt_P1} pointwise over space, we obtain the optimal \ac{lra} probability and the corresponding transmission probability, as shown in Figs.~\ref{fig:geofencing}c and~\ref{fig:geofencing}d.

Based on the resulting $\bar{J}_{\mathrm{LRA}}^{\mathrm{TS}*}$ map, we further define a geofencing region by imposing the constraint $\bar{J}_{\mathrm{LRA}}^{\mathrm{TS}*}<\tau$, where $\tau$ is a prescribed {\ac{lra}-probability} threshold. The resulting geofencing boundary is shown by the black contour in Figs.~\ref{fig:geofencing}c and~\ref{fig:geofencing}d. This contour bounds the region in which even the optimal RS transmission policy cannot meet the prescribed \ac{lra} requirement under the assumed \ac{zoh} reconstruction model. It is also observed that the geofencing region becomes relatively large in the left and right directions, where the lack of blockage leads to more favorable propagation for Eve. This observation motivates further study on how to shrink the geofencing region, for example by introducing additional blockages or by designing proactive interference or jamming strategies.

\section{Conclusion}
This work developed an analytical framework for legitimate reconstruction advantage in semantics-aware remote monitoring under eavesdropping. The proposed LRA metrics characterize the probability and persistence of correct reconstruction at the intended receiver coupled with incorrect reconstruction at the eavesdropper. A joint stationary analysis yielded closed-form expressions for general finite-state ergodic Markov sources, with explicit algebraic specializations enabling efficient transmission-probability optimization. The results show that achieving legitimate advantage more frequently and sustaining it can favor different source dynamics and transmission policies. For two-state sources, rare-transmission analysis reveals that greater transition asymmetry can increase persistence while reducing probability when the sum of transition probabilities is fixed. Moreover, optimizing receiver-wise objectives can select policies that are suboptimal for LRA, and improvements in legitimate reconstruction accuracy need not increase the probability of legitimate advantage. These findings emphasize that transmission decisions should account for both receivers' reconstruction outcomes and the temporal requirements of the monitoring application. The geofencing example further translates prescribed LRA requirements into location-dependent spatial boundaries.
These conclusions hold under randomized stationary transmission and prescribed ZOH reconstruction at both receivers. An important direction for future work is to characterize how source-aware adversarial inference changes the preferred transmission policies and the resulting spatial boundaries.
\appendices{}

\section{Proof of Lemma~\ref{lem:mc_3d}}\label{sec:proof_mc_3d}
{Since the state space $\mathcal N^3$ is finite, $\{S_t\}_{t\in\mathcal T}$ is a finite-state process.}
{By the Markov property of the source and its independence from the transmission and channel-state processes, the conditional distribution of $X_{t+1}$ given $(S_1,\ldots,S_t)$ depends only on $X_t$. }Moreover, conditioned on $S_t$ and $X_{t+1}$, the state reconstruction at receiver $r\in\mathcal R$ evolves according to
\begin{equation*} 
\hat{X}_{t+1}^r= 
\begin{cases} 
X_{t+1}, & A_{t+1} H_{t+1}^r=1,\\ 
\hat{X}_{t}^r,  & A_{t+1} H_{t+1}^r=0, \end{cases}
\end{equation*} 
{under the \ac{zoh} reconstruction rule in Section~\ref{subsec:remote_estimation}. The random vector $(A_{t+1},H_{t+1}^{b},H_{t+1}^{e})$ is independent of $(S_1,\ldots,S_t,X_{t+1})$.} Consequently, conditioned on $S_t$, the distribution of $S_{t+1}$ is independent of all earlier states. Therefore,
\begin{equation*}
\Pr[S_{t+1}\mid S_1,\ldots,S_t]
=
\Pr[S_{t+1}\mid S_t],
\end{equation*}
which establishes the Markov property of $\{S_t\}_{t\in\mathcal T}$.

For any
$\mathbf{s}=(x,a,b)\in\mathcal N^3$ and $\mathbf{s}'=(x',a',b')\in\mathcal N^3$, define the one-step transition probability as
$
\mathbf{P}_{\mathbf{s},\mathbf{s}'}
\triangleq
\Pr[S_{t+1}=\mathbf{s}'\mid S_t=\mathbf{s}].
$
Conditioning on the next source state gives
\begin{align*}
\mathbf{P}_{\mathbf{s},\mathbf{s}'}
&=
\Pr[X_{t+1}=x'\mid X_t=x]\nonumber\\
&\quad\times
\Pr[
\hat X_{t+1}^{\mathrm b}=a',
\hat X_{t+1}^{\mathrm e}=b'
\mid S_t=\mathbf{s},X_{t+1}=x']\\
& = [\mathbf Q]_{x,x'} \Psi(\mathbf{s},\mathbf{s}').
\end{align*}
This establishes \eqref{eq:kernel}.

Finally, since $\mathbf Q$ and $\Psi(\mathbf{s},\mathbf{s}')$ are independent of time $t$, the induced process $\{S_t\}_{t\in\mathcal T}$ is time-homogeneous.

\section{Proof of Lemma~\ref{lem:ergodicity}}\label{sec:proof_ergodicity}
Since $\{S_t\}_{t\in\mathcal T}$ is a finite-state, time-homogeneous Markov chain, it suffices to establish {the existence of a unique recurrent class and its aperiodicity}.
These properties imply ergodicity {on that class} and the existence of a unique stationary distribution.

\subsubsection{Uniqueness of the Recurrent Class}
Consider an arbitrary state $\mathbf{s}=(x,a,b)$ and a diagonal state $\mathbf{s}'=(x',x',x')$ in $\mathcal N^3$. Since the source Markov chain is irreducible, there exists a positive integer $n$ such that $[\mathbf{Q}^{n}]_{x,x'}>0$. Let the source move from $x$ to $x'$ in $n$ slots. Suppose that both receivers update in the final slot. From the transition matrix in Lemma~\ref{lem:mc_3d}, it follows that
\begin{equation*}
[\mathbf P^n]_{s,s'}
\ge
[\mathbf Q^n]_{x,x'}\lambda_{11}>0,
\end{equation*}
where $\lambda_{11}=p_{\mathrm a}p_{\mathrm h}^{b}p_{\mathrm h}^{e}>0$ under the assumptions in Section~\ref{subsec:remote_estimation}. Accordingly, $\mathbf{s}'$ is reachable from $\mathbf{s}$ with positive probability. Since $\mathbf{s}$ is arbitrary, the diagonal state $\mathbf{s}'$ is reachable from every joint state.

Since the chain is finite, at least one recurrent class exists. Any such class must contain $\mathbf{s}'$, because $\mathbf{s}'$ is reachable from every state, and a recurrent class cannot be left once entered. If there were multiple recurrent classes, they would therefore share $\mathbf{s}'$, contradicting their disjointness. Hence, the recurrent class is unique.

\subsubsection{Aperiodicity}
For any source state $x\in\mathcal N$, define the set of positive-probability return times to source state $x$ as
\begin{equation*}
    \mathcal R_\mathbf{Q}(x)
    \triangleq
    \{n\geq 1:[\mathbf{Q}^n]_{x,x}>0\}.
\end{equation*}
Since the source Markov chain is aperiodic, the greatest common divisor of the return times in $\mathcal R_\mathbf{Q}(x)$ is 1, \emph{i.e.},
$
    \gcd \mathcal R_{\mathbf Q}(x)=1,\, \forall x\in\mathcal N .
$

Consider {a diagonal state $\mathbf s=(x,x,x)\in\mathcal N^3$, which belongs to the unique recurrent class}. Let
\begin{equation*}
    \mathcal R_{\mathbf P}(\mathbf s)
    \triangleq
    \{n\geq 1:[\mathbf P^n]_{\mathbf s,\mathbf s}>0\}
\end{equation*}
denote the return-time set of $\mathbf s$ under the induced joint transition matrix
$\mathbf P$.

For every $n\in\mathcal R_\mathbf{Q}(x)$, an $n$-slot return of the source combined with {a simultaneous reception at both receivers in the final slot} returns the composite state $\mathbf s =(x,x,x)$ to itself. Hence,
\begin{equation*}
    [\mathbf P^n]_{\mathbf s,\mathbf s}
    \geq
    [\mathbf Q^n]_{x,x}\lambda_{11}
    >0 ,
\end{equation*}
which implies
$
    \mathcal R_\mathbf Q(x)\subseteq \mathcal R_{\mathbf P}(\mathbf s).
$
Combining with $\gcd \mathcal R_\mathbf Q(x)=1$, the return-time set $\mathcal R_{\mathbf P}(\mathbf s)$ also has greatest common divisor one, \emph{i.e.},
$
    \gcd \mathcal R_{\mathbf P}(\mathbf s)=1.
$
{Since all states in a communicating class have the same period, the unique recurrent class is aperiodic.}

{Therefore, the induced process is ergodic on its unique recurrent class and admits a unique stationary distribution, to which its distribution converges from any initial state.}

\section{Proof of Theorem~\ref{thm:stationary}} \label{sec:proof_thm_stationary}
We start by conditioning on the joint age process $(\Theta^\mathrm{b}_t,\Theta^\mathrm{e}_t)\in\mathbb Z_+^2$. By the law of total probability, we have
\begin{align}
    \pi(\mathbf{s})
    &=
    \lim_{t\to\infty}
    \sum_{i=0}^\infty \sum_{j=0}^\infty
     \Pr(
        \Theta_t^{\mathrm b}=i,
        \Theta_t^{\mathrm e}=j
    ) \nonumber \\
    & \times
    \Pr(
        X_t=x,\,
        \hat X_t^{\mathrm b}=a,
        \hat X_t^{\mathrm e}=b
        \mid
        \Theta_t^{\mathrm b}=i,
        \Theta_t^{\mathrm e}=j
    ).
    \label{eq:stationary_c_age}
\end{align}

Under the state reconstruction model induced by the age variables, as defined in  \eqref{eq:def_estimator_receiver}, the second term in \eqref{eq:stationary_c_age} becomes
\begin{align*}
    &\Pr(
        X_t=x,\,
        \hat X_t^{\mathrm b}=a, \hat X_t^{\mathrm e}=b  \mid  \Theta_t^{\mathrm b}=i, \Theta_t^{\mathrm e}=j)\\
    & \quad = \Pr(X_t=x, X_{t-i}=a, X_{t-j}=b \mid \Theta_t^{\mathrm b}=i, \Theta_t^{\mathrm e}=j)\\
    & \quad = \Pr(X_t=x, X_{t-i}=a, X_{t-j}=b),
\end{align*}
where the last equality holds because the age process is fully determined by $\{A_t,H^\mathrm{b}_t,H^\mathrm{e}_t\}_{t\in\mathcal T}$, as defined in \eqref{eq:def_aoi_receiver} and \eqref{eq:def_estimator_receiver}, and is independent of the source process $\{X_t\}_{t\in\mathcal T}$.

Since $\{X_t\}_{t\in\mathcal T}$ is a finite-state ergodic Markov chain, and the joint age process $(\Theta^\mathrm{b}_t,\Theta^\mathrm{e}_t)$ {has a unique positive-recurrent, aperiodic class containing} $(0,0)$ due to $\lambda_{11}>0$, the following limits exist.
\begin{equation*}
    \rho(i,j) \triangleq \lim_{t\to\infty} \Pr (\Theta_t^{\mathrm b}=i, \Theta_t^{\mathrm e}=j),
\end{equation*}
\begin{equation*}
    \phi(x,a,b,i,j) \triangleq \lim_{t\to\infty} \Pr(X_t=x, X_{t-i}=a, X_{t-j}=b).
\end{equation*}
Thus, the joint stationary probability becomes
\begin{equation}
    \pi(\mathbf{s}) = \sum_{i=0}^\infty \sum_{j=0}^\infty \rho(i,j)\phi(x,a,b,i,j). \label{eq:stationary_d_limit}
\end{equation}

Next, we first derive the limiting distributions $\phi(x,a,b,i,j)$ and $\rho(i,j)$ separately, and then combine them to obtain the joint stationary distribution $\pi(\mathbf{s})$.

\subsubsection{Stationary Distribution of Joint Age Process}
The stationary distribution of the joint age process $\rho(i,j)$ is characterized by three cases.

\textit{Case 1: $i = j$ (Synchronized updates).}
This case corresponds to the event that the most recent successful update is received simultaneously by both Bob and Eve. The age pair is reset to $(0,0)$ with probability $\lambda_{11}$, and then increases synchronously only if neither receiver successfully receives a new update in the subsequent slots. Since a joint failure occurs with probability $\lambda_{00}$ in each slot, we obtain
\begin{equation}
    \rho(i,i)  = \lambda_{11} (\lambda_{00})^i,\,i\ge0. \label{eq:stationary_a_1}
\end{equation}

\textit{Case 2: $i < j$ (Bob is fresher).}
Let $\Delta = j - i > 0$. Then, at time $t-i$, the system must be in a boundary state $(0, \Delta)$, meaning that Bob has just received a successful update while Eve’s age exceeds Bob’s by $\Delta$. From that time to time $t$, the age pair can evolve from $(0, \Delta)$ to $(i,j)$ only if $i$ consecutive joint failures occur. Therefore, 
\begin{equation}
    \rho(i,j) = \rho(0, \Delta) (\lambda_{00})^i.
    \label{eq:mid_age_pi}
\end{equation}
It remains to determine the boundary term $\rho(0, \Delta)$. The event $(\Theta_{t-i}^\mathrm{b},\Theta_{t-i}^\mathrm{e})=(0, \Delta)$ occurs when Bob successfully receives an update while Eve does not, and Eve’s age in the previous slot equals $\Delta-1$. The former occurs with probability $\lambda_{10}$ {independently of the previous age}. Moreover, Eve successfully receives an update in each slot with marginal probability $p_\mathrm{e}=\lambda_{11}+\lambda_{01}$, so its marginal age distribution is geometric 
\begin{equation*}
    \Pr(\Theta^e = k) = p_\mathrm{e} (1-p_\mathrm{e})^k,\, k\ge 0
\end{equation*}
Hence,
\begin{equation*}
    \rho(0, \Delta) = \lambda_{10}  p_\mathrm{e} (1 - p_\mathrm{e})^{\Delta - 1} .
\end{equation*}
Substituting $\rho(0, \Delta)$ into \eqref{eq:mid_age_pi}, we have 
\begin{equation}
    \rho(i,j) = \lambda_{10} p_\mathrm{e} (1 - p_\mathrm{e})^{j - i - 1} (\lambda_{00})^i .\label{eq:stationary_a_2}
\end{equation}

\textit{Case 3: $i > j$ (Eve is fresher).}
This case is symmetric to Case 2, with Bob and Eve interchanged. Hence,
\begin{equation}
    \rho(i,j) = \lambda_{01} p_\mathrm{b}  (1 - p_\mathrm{b} )^{i - j- 1} (\lambda_{00})^j, \label{eq:stationary_a_3}
\end{equation}
where $p_\mathrm{b} = \lambda_{11} + \lambda_{10}$ denotes the marginal success probability at Bob.

\subsubsection{Stationary Three-Time Joint Distribution}
The stationary three-time joint distribution $\phi(x,a,b,i,j)$ is also characterized by three cases.

\textit{Case 1: $i = j$ (Synchronized updates).}
In this case, the time instants $t-i$ and $t-j$ coincide. Hence, the event has nonzero probability only if $a=b$. Under stationarity, the source is in state $a$ with probability $[\mathbf v]_a$, and then evolves from $a$ to $x$ over $i$ steps according to $\mathbf Q^i$. Therefore,
\begin{equation}
    \phi(x,a,b,i,i) = \mathbbm{1}(a=b) [\mathbf v]_a [\mathbf Q^i]_{a,x}. \label{eq:stationary_3_1}
\end{equation}

\textit{Case 2: $i < j$ (Bob is fresher).}
Here, $t-j<t-i\le t$. The process follows the path $\xrightarrow{\infty}b \xrightarrow{j-i} a \xrightarrow{i} x$. By the Markov property and stationarity, we have
\begin{equation}
    \phi(x,a,b,i,j) = [\mathbf v]_b [\mathbf Q^{j-i}]_{b,a} [\mathbf Q^i]_{a,x}. \label{eq:stationary_3_2}
\end{equation}

\textit{Case 3: $i > j$ (Eve is fresher).}
Here, $t-i<t-j\le t$. The process follows the path $\xrightarrow{\infty}a \xrightarrow{i-j} b \xrightarrow{j} x$. Similarly, we have
\begin{equation}
    \phi(x,a,b,i,j) = [\mathbf v]_a [\mathbf Q^{i-j}]_{a,b} [\mathbf Q^j]_{b,x}. \label{eq:stationary_3_3}
\end{equation}

\subsubsection{Stationary Distribution of Induced Process}
Having characterized the stationary age distribution $\rho(i,j)$ and the stationary three-time joint distribution, we now derive the stationary distribution $\pi(\mathbf{s})$. As above, the derivation is carried out by distinguishing three cases.

\textit{Case 1: $i = j$ (Synchronized updates).}
In this case, the two receivers share the same freshest update, and therefore $a=b$. Substituting \eqref{eq:stationary_a_1} and \eqref{eq:stationary_3_1} into \eqref{eq:stationary_d_limit}, we have
\begin{align*}
    \pi_1(\mathbf{s}) &\triangleq \sum_{i=0}^\infty \mathbbm{1}(a=b) [\mathbf v]_a [\mathbf Q^i]_{a,x} \lambda_{11} (\lambda_{00})^i \nonumber \\
    &= \mathbbm{1}(a=b) \lambda_{11} [\mathbf v]_a \left[ \sum_{i=0}^\infty (\lambda_{00} \mathbf Q)^i \right]_{a,x} \nonumber \\
    &= \mathbbm{1}(a=b) \lambda_{11} [\mathbf v]_a \left[ (\mathbf{I}_N - \lambda_{00}\mathbf Q)^{-1} \right]_{a,x}.
\end{align*}

\textit{Case 2: $i < j$ (Bob is fresher).}
Let $\Delta = j - i \ge 1$. Substituting \eqref{eq:stationary_a_2} and \eqref{eq:stationary_3_2} into \eqref{eq:stationary_d_limit}, we have
\begin{align*}
    & \pi_2(\mathbf{s})\nonumber\\ 
    &\triangleq \sum_{i=0}^\infty \sum_{\Delta=1}^\infty [\mathbf v]_b [\mathbf Q^\Delta]_{b,a} [\mathbf Q^i]_{a,x}  \lambda_{10} p_\mathrm{e} (1 - p_\mathrm{e})^{\Delta - 1} (\lambda_{00})^i \nonumber \\
    &= [\mathbf v]_b \lambda_{10} p_\mathrm{e} \left[ \sum_{\Delta=1}^\infty (1-p_\mathrm{e})^{\Delta-1} \mathbf Q^\Delta \right]_{b,a} \left[ \sum_{i=0}^\infty (\lambda_{00}\mathbf Q)^i \right]_{a,x} \nonumber \\
    & = [\mathbf v]_b \lambda_{10} p_\mathrm{e} \left[ \mathbf Q(\mathbf{I}_N  - (1-p_\mathrm{e})\mathbf Q)^{-1} \right]_{b,a} \left[ (\mathbf{I}_N  - \lambda_{00}\mathbf Q)^{-1} \right]_{a,x}.
\end{align*}

\textit{Case 3: $i > j$ (Eve is fresher).}
By symmetry, let $\Delta = i - j \ge 1$. Substituting \eqref{eq:stationary_a_3} and \eqref{eq:stationary_3_3} into \eqref{eq:stationary_d_limit}, we have
\begin{align*}
    \pi_3(\mathbf{s}) &= [\mathbf{v}]_a \lambda_{01} p_\mathrm{b} \left[ \mathbf Q(\mathbf I_N - (1-p_\mathrm{b})\mathbf Q)^{-1} \right]_{a,b} \nonumber \\
    &\quad \times \left[ (\mathbf I - \lambda_{00}\mathbf Q)^{-1} \right]_{b,x}.
\end{align*}

Combining the three cases above yields \eqref{eq:stationary_distribution}.

\section{Proof of Theorem~\ref{thm:stationary_metrics} }\label{sec:proof_stationary_metrics}
The \ac{lra} event occurs in slot $t$ if and only if $S_t\in\mathcal S_\mathrm{LRA}$. Hence, by the definition of $\bar J_\mathrm{LRA}$ in \eqref{eq:avg_reconstruction_advantage},
\begin{equation*}
    \bar J_\mathrm{LRA} \stackrel{(a)}{=}
    \sum_{\mathbf s\in\mathcal S_\mathrm{LRA}}
    \pi(\mathbf s)
    =
    \mathbf 1_{N_\mathrm{LRA}}^{\mathsf T}\boldsymbol\pi_\mathrm{LRA}.
\end{equation*}
Here, $(a)$ follows from Lemma~\ref{lem:ergodicity}.

Next, consider the persistence metric. Since $D_\mathrm{LRA}[t]$ is nonnegative
integer-valued, it can be rewritten as
\begin{align}
    D_\mathrm{LRA}[t]
     &=     \sum_{d=1}^{t}
    \mathbbm 1\{D_\mathrm{LRA}[t]\geq d\}\nonumber\\
     &= \sum_{d=1}^{t}
    \mathbbm 1\{
    S_{t-d+1},\ldots,S_t\in\mathcal S_\mathrm{LRA}
    \}.
    \label{eq:D_xi_sum_form}
\end{align}
The second equation holds because $D_\mathrm{LRA}[t]\geq d$ if and only if the \ac{lra} event holds throughout the last $d$ slots.
Substituting \eqref{eq:D_xi_sum_form} into the definition of $\bar D_\mathrm{LRA}$ gives
\begin{align}
    \bar D_\mathrm{LRA}
    &=
    \liminf_{T\to\infty}
    \frac{1}{T}
    \sum_{t=1}^{T}
    \sum_{d=1}^{t}
    \Pr\{
        S_{t-d+1},\ldots,S_t\in\mathcal S_\mathrm{LRA}
    \}.
    \label{eq:proof_D_definition}
\end{align}
For any fixed $d$, the Markov property gives
\begin{align*}
    &\Pr\{
        S_{t-d+1},\ldots,S_t\in\mathcal S_\mathrm{LRA}
    \}
    \nonumber\\
    &=
    \sum_{\mathbf s_0,\ldots,\mathbf s_{d-1}\in\mathcal S_\mathrm{LRA}}
    \Pr\{
    S_{t-d+1}=\mathbf s_0,\ldots,S_t=\mathbf s_{d-1}
    \}
    \nonumber\\
    &=
    \sum_{\mathbf s_0,\ldots,\mathbf s_{d-1}\in\mathcal S_\mathrm{LRA}}
    \Pr\{S_{t-d+1}=\mathbf s_0\}
    \mathbf P_{\mathbf s_0,\mathbf s_1}
    \ldots
    \mathbf P_{\mathbf s_{d-2},\mathbf s_{d-1}},
    \label{eq:proof_D_path_expansion}
\end{align*}
{By Lemma~\ref{lem:ergodicity}, for each fixed $d$, the time-averaged distribution of $S_{t-d+1}$ converges to the stationary distribution. Moreover, an \ac{lra} run cannot continue through a successful reception at Eve, so
\begin{equation*}
    \Pr\{S_{t-d+1},\ldots,S_t\in\mathcal S_{\mathrm{LRA}}\}
    \le (1-p_e)^{d-1},
    \, 1\le d\le t.
\end{equation*}
Since $p_e>0$, this bound is summable over $d$ and justifies interchanging the limit and the sum in~\eqref{eq:proof_D_definition}.} Hence,
\begin{align*}
    \bar D_\mathrm{LRA}
    &=
    \sum_{d=1}^{\infty}
    \sum_{\mathbf s_0,\ldots,\mathbf s_{d-1}\in\mathcal S_\mathrm{LRA}}
    \pi(\mathbf s_0)
    \mathbf P_{\mathbf s_0,\mathbf s_1}
    \ldots
    \mathbf P_{\mathbf s_{d-2},\mathbf s_{d-1}}
    \nonumber\\
    &=
    \sum_{d=1}^{\infty}
    \mathbf 1_{N_\mathrm{LRA}}^{\mathsf T}
    \left(\mathbf P_\mathrm{LRA}^{\mathsf T}\right)^{d-1}
    \boldsymbol\pi_\mathrm{LRA}.
\end{align*}
{Similarly, since a successful reception at Eve precludes LRA, each row sum of $\mathbf P_{\mathrm{LRA}}$ is at most $1-p_e$. Thus,
\begin{equation*}
\rho(\mathbf P_{\mathrm{LRA}})
\le
\|\mathbf P_{\mathrm{LRA}}\|_{\infty}
\le
1-p_e<1,
\end{equation*}}
and the Neumann series gives 
$
    \sum_{d=0}^{\infty}(\mathbf P_\mathrm{LRA}^{\mathsf T})^d
    =
    (\mathbf I_{N_\mathrm{LRA}}-\mathbf P_\mathrm{LRA}^{\mathsf T})^{-1}.
$
Therefore,
\begin{equation*}
    \bar D_\mathrm{LRA}
    =
    \mathbf 1_{N_\mathrm{LRA}}^{\mathsf T}
    \left(\mathbf I_{N_\mathrm{LRA}}-\mathbf P_\mathrm{LRA}^{\mathsf T}\right)^{-1}
    \boldsymbol\pi_\mathrm{LRA}.
\end{equation*}

\section{Proof of Corollary~\ref{cor:lra_probability_ut}}\label{sec:proof_lra_probability_ut}
The transition matrix $\mathbf Q^\mathrm{UT}$ in \eqref{eq:uniform_transition_Q} can be written as
\begin{equation*}
    \mathbf Q^{\mathrm{UT}}
    =
    d\mathbf I_N+\eta\mathbf 1_N\mathbf 1_N^{\mathsf T},
\end{equation*}
where $d= \rho-\eta=1-N\eta$. For any $c\in[0,1)$, the Sherman-Morrison formula gives
\begin{equation*}
    \left(\mathbf I_N-c\mathbf Q^{\mathrm{UT}}\right)^{-1}
    =
    \frac{1}{1-cd}\mathbf I_N
    +
    \frac{c\eta}{(1-cd)(1-c)}
    \mathbf 1_N\mathbf 1_N^{\mathsf T},
\end{equation*}
and
\begin{equation*}
\mathbf Q^{\mathrm{UT}}\left(\mathbf I_N-c\mathbf Q^{\mathrm{UT}}\right)^{-1}
=
\frac{d}{1-cd}\mathbf I_N
+
\frac{\eta}{(1-cd)(1-c)}
\mathbf 1_N\mathbf 1_N^{\mathsf T}.
\end{equation*}

Setting $c=\lambda_{00}$, denote the common diagonal and off-diagonal entries of $\left(\mathbf I_N-c\mathbf Q^{\mathrm{UT}}\right)^{-1}$ by
    \begin{equation*}
        \psi_\mathrm{d} = \frac{1+\lambda_{00}(\eta-1)}{(1-\lambda_{00}d)(1-\lambda_{00})},\,
        \psi_\mathrm{o}  = \frac{\lambda_{00}\eta}{(1-\lambda_{00}d)(1-\lambda_{00})}.
    \end{equation*}

The \ac{lra} event corresponds to $a=x$ and $b\neq x$. Therefore, for every
$b\neq x$, by Theorem~\ref{thm:stationary}, we have
\begin{align*}
    \pi_1^{\mathrm{UT}}(x,x,b) &= 0,\,
    \pi_2^{\mathrm{UT}}(x,x,b) = \frac{1}{N}
    \frac{\eta\lambda_{10}}{1-(1-p_{\mathrm e})d}
    \psi_\mathrm{d}, \nonumber \\
    \pi_3^{\mathrm{UT}}(x,x,b) &= \frac{1}{N}
    \frac{\eta\lambda_{01}}{1-(1-p_{\mathrm b})d}
    \psi_{\mathrm o}.
\end{align*}
{These probabilities are identical for all $N(N-1)$ \ac{lra} states.} Therefore, the long-term average \ac{lra} probability is 
\begin{align*}
    \bar J_{\mathrm{LRA}}^{\mathrm{UT}}
    &=
    \sum_{x\in\mathcal N ^\mathrm{UT}}
    \sum_{b\in\mathcal N^\mathrm{UT}\setminus\{x\}}
    \pi(x,x,b)
    =N(N-1)\pi(x,x,b)\nonumber \\
    &=\frac{(N-1)\eta\lambda_{10}\psi_{\mathrm d}}{1-(1-p_\mathrm{e})(\rho-\eta)}
        +
        \frac{(N-1)\eta\lambda_{01}\psi_{\mathrm o}}{1-(1-p_\mathrm{b})(\rho-\eta)}.
\end{align*}

\section{Proof of Corollary~\ref{cor:lra_persistence_ut}}\label{sec:proof_ut_lra_persistence}
Recall that $\mathcal S_{\mathrm{LRA}}^{\mathrm{UT}}$ contains $N_\mathrm{LRA}=N(N-1)$ states of the form $(x,x,b)$ with $b\neq x$.
By Lemma~\ref{lem:mc_3d}, the transition probability from $\mathbf s=(x,x,b)\in\mathcal S_{\mathrm{LRA}}^{\mathrm{UT}}$ to $\mathbf s'=(x',x',b')\in\mathcal S_{\mathrm{LRA}}^{\mathrm{UT}}$ is
\begin{equation} 
[\mathbf P_{\mathrm{LRA}}^{\mathrm{UT}}]_{s,s'} 
= 
\begin{cases} 
    \rho(1-p_{\mathrm e}), & x'=x,\ b'=b,\\ 
    \eta\lambda_{10}, & x'\in\mathcal N\setminus\{x,b\},\ b'=b,\\ 
    0, & \text{otherwise}. 
\end{cases} 
\label{eq:proof_ut_lra_transition_matrix} 
\end{equation}
{A successful reception at Eve makes its reconstruction correct and therefore precludes \ac{lra}. Hence, any transition within $\mathcal S_{\mathrm{LRA}}^{\mathrm{UT}}$ requires Eve to receive no update, so $b'=b$.}

It follows from \eqref{eq:proof_ut_lra_transition_matrix} that every row of
$\mathbf P_{\mathrm{LRA}}^{\mathrm{UT}}$ has the same row sum
$
    \kappa_{\mathrm{LRA}}^{\mathrm{UT}}
    =
    \rho(1-p_{\mathrm e})+(N-2)\eta\lambda_{10}.
    \label{eq:proof_ut_lra_rho}
$
Equivalently,
\begin{equation*}
    \mathbf P_{\mathrm{LRA}}^{\mathrm{UT}}
    \mathbf 1_{N_{\mathrm{LRA}}}
    =
    \kappa_{\mathrm{LRA}}^{\mathrm{UT}}
    \mathbf 1_{N_{\mathrm{LRA}}}.
\end{equation*}
Thus,
\begin{equation*}
    \mathbf 1_{N_{\mathrm{LRA}}}^{\mathsf T}
    \left(
    \mathbf I_{N_{\mathrm{LRA}}}
    -
    \left(\mathbf P_{\mathrm{LRA}}^{\mathrm{UT}}\right)^{\mathsf T}
    \right)
    =
    \left(1-\kappa_{\mathrm{LRA}}^{\mathrm{UT}}\right)
    \mathbf 1_{N_{\mathrm{LRA}}}^{\mathsf T}.
\end{equation*}
Multiplying both sides by
$\left(
\mathbf I_{N_{\mathrm{LRA}}}
-
(\mathbf P_{\mathrm{LRA}}^{\mathrm{UT}})^{\mathsf T}
\right)^{-1}$
from the right yields
\begin{equation*}
    \mathbf 1_{N_{\mathrm{LRA}}}^{\mathsf T}
    \left(
    \mathbf I_{N_{\mathrm{LRA}}}
    -
    \left(\mathbf P_{\mathrm{LRA}}^{\mathrm{UT}}\right)^{\mathsf T}
    \right)^{-1}
    =
    \frac{1}{1-\kappa_{\mathrm{LRA}}^{\mathrm{UT}}}
    \mathbf 1_{N_{\mathrm{LRA}}}^{\mathsf T}.
\end{equation*}

By Theorem~\ref{thm:stationary_metrics}, the long-term average \ac{lra} persistence is
\begin{align*}
    \bar D_{\mathrm{LRA}}^{\mathrm{UT}}
    &=\mathbf 1_{N_{\mathrm{LRA}}}^{\mathsf T}
    \left(
    \mathbf I_{N_{\mathrm{LRA}}}
    -
    \left(\mathbf P_{\mathrm{LRA}}^{\mathrm{UT}}\right)^{\mathsf T}
    \right)^{-1}
    \boldsymbol\pi_{\mathrm{LRA}}^{\mathrm{UT}}\\
    &=
    \frac{1}{1-\kappa_{\mathrm{LRA}}^{\mathrm{UT}}}
    \mathbf 1_{N_{\mathrm{LRA}}}^{\mathsf T}
    \boldsymbol\pi_{\mathrm{LRA}}^{\mathrm{UT}}
    =
    \frac{
    \bar J_{\mathrm{LRA}}^{\mathrm{UT}}
    }{
    1-\kappa_{\mathrm{LRA}}^{\mathrm{UT}}
    }\\
    &=
    \frac{
    \bar J_{\mathrm{LRA}}^{\mathrm{UT}}
    }{
    1-\rho(1-p_{\mathrm e})-(N-2)\eta\lambda_{10}
    }.
\end{align*}

\section{Proof of Corollary~\ref{cor:opt_P1}} \label{sec:proof_cor_opt_P1}
Before proving Corollary~\ref{cor:opt_P1}, we first characterize the signs of the coefficients.
{Under the model assumptions, we have} that $\beta_2\le 0$, $\beta_3\le 0$, and $\beta_5\le 0$. In contrast, the signs of $\beta_1$  and $\beta_4$ are parameter-dependent and cannot be determined a priori. Next, to determine the optimal $p_\mathrm{a}$, we analyze the shape of $\bar{J}_{\mathrm{LRA}}^\mathrm{TS}$ through the sign of its derivative.

From Theorem~\ref{thm:polynomial_optimal_pa}, {a maximizer of the continuously extended objective} $\bar{J}_{\mathrm{LRA}}^\mathrm{TS}$ belongs to 
\begin{equation*}
    \mathcal C_F
    \triangleq
    \{0,1\}
    \cup
    \left\{
    p_{\mathrm a}\in(0,1):
    \mathcal P_F(p_{\mathrm a})=0
    \right\},
\end{equation*}
where 
$
\mathcal P_F(p_{\mathrm a})=-\beta_1\beta_3\,p_\mathrm{a}^2-2\beta_2 \beta_3\,p_\mathrm{a}+(\beta_1\beta_5-\beta_2\beta_4).
$

If $\beta_1\beta_3\neq0$, the zero points of $\mathcal P_F(p_{\mathrm a})$ are given by
\begin{equation*}
r_{\pm}=\frac{-\beta_2\beta_3\pm\sqrt{\Delta}}{\beta_1\beta_3},
\ 
\Delta\triangleq\beta_2^2\beta_3^2+\beta_1\beta_3(\beta_1\beta_5-\beta_2\beta_4).
\end{equation*}
Substituting the explicit expressions of $\beta_1$, $\beta_2$, $\beta_3$, $\beta_4$, and $\beta_5$ into $\Delta$, we obtain
\begin{equation*}
\Delta
=
p^2 q^2 p_\mathrm{h}^\mathrm{b} p_\mathrm{h}^\mathrm{e} (p+q)^4 (p+q-1)^2 (p_\mathrm{h}^\mathrm{b}-p_\mathrm{h}^\mathrm{e})^2
(p_\mathrm{h}^\mathrm{b}+p_\mathrm{h}^\mathrm{e}-p_\mathrm{h}^\mathrm{b} p_\mathrm{h}^\mathrm{e})^2
\end{equation*}
\begin{equation*}
\times
\left[p_\mathrm{h}^\mathrm{b}+(p+q)p_\mathrm{h}^\mathrm{e} (1-p_\mathrm{h}^\mathrm{b})\right]
\left[p_\mathrm{h}^\mathrm{e}+p_\mathrm{h}^\mathrm{b} (2-p-q)(1-p_\mathrm{h}^\mathrm{e})\right].
\end{equation*}
Since $p,q,p_\mathrm{h}^\mathrm{b},p_\mathrm{h}^\mathrm{e}\in[0,1]$ and $p+q\in[0,2]$, every factor in the above expression is nonnegative. Hence, $\Delta\ge 0$.

Next, differentiating $\mathcal P_F(p_\mathrm{a})$ yields
\begin{equation*}
\frac{d \mathcal P_F}{d p_\mathrm{a}}=-2\beta_1\beta_3 p_\mathrm{a}-2\beta_2\beta_3=-2\beta_3(\beta_1 p_\mathrm{a}+\beta_2),
\end{equation*}
and thus the only vertex of $\mathcal P_F$ is located at $p_\mathrm{a}^{\mathrm v}=-\frac{\beta_2}{\beta_1}$.
We next distinguish two cases based on the sign of $p_\mathrm{a}^{\mathrm v}$.

\textit{Case 1: $\beta_1<0$.}
Since $\beta_2\le 0$, the vertex of $\mathcal P_F(p_\mathrm{a})$, $p_\mathrm{a}^{\mathrm v}=-\beta_2/\beta_1\le 0$.
Moreover, $\beta_3\le 0$ and $\beta_1<0$ imply that $\mathcal P_F(p_{\mathrm a})$ is concave. Therefore, $\mathcal P_F(p_{\mathrm a})$ is decreasing on $[0,1]$.

Since $\beta_1\beta_3\ge0$, the two roots satisfy
\begin{equation*}
r_\mathrm{s}=\frac{-\beta_2\beta_3-\sqrt{\Delta}}{\beta_1\beta_3}
\le
\frac{-\beta_2\beta_3+\sqrt{\Delta}}{\beta_1\beta_3}=r_\mathrm{l}.
\end{equation*}
Because the vertex lies to the left of $[0,1]$, the smaller root satisfies $r_\mathrm{s}<0$ and only the larger root $r_\mathrm{l}$ may lie in $[0,1]$. Accordingly, $\bar{J}_{\mathrm{LRA}}^\mathrm{TS}$ is increasing, decreasing, or increasing-then-decreasing on $[0,1]$ depending on whether $r_\mathrm{l}\ge 1$, $r_\mathrm{l}\le 0$, or $r_\mathrm{l}\in(0,1)$, respectively. Therefore, $p_\mathrm{a}^\dagger$ is obtained by projecting $r_\mathrm{l}$ onto $[0,1]$, namely,
\begin{equation*}
p_\mathrm{a}^\dagger=\min\left\{1,\max\left\{\frac{-\beta_2\beta_3+\sqrt{\Delta}}{\beta_1\beta_3},0\right\}\right\}.
\end{equation*}

\textit{Case 2: $\beta_1>0$.}
Using the explicit expressions of $\beta_1$ and $\beta_2$, we obtain
\begin{equation*}
\beta_1+\beta_2
=
-pq(1-p_\mathrm{h}^\mathrm{e})\left[(p+q)(1-p_\mathrm{h}^\mathrm{b})(p_\mathrm{h}^\mathrm{b}+p_\mathrm{h}^\mathrm{e})+2(p_\mathrm{h}^\mathrm{b})^2\right],
\end{equation*}
which is less than or equal to zero, that is, $\beta_2\le -\beta_1$. Since $\beta_1>0$, it follows that the vertex of $\mathcal P_F(p_{\mathrm a})$, $
p_{\mathrm a}^{\mathrm v}=-{\beta_2}/{\beta_1}\ge 1$. Moreover, $\beta_3\le 0$ and $\beta_1>0$ imply that $\mathcal P_F(p_{\mathrm a})$ is convex. Therefore $\mathcal P_F(p_{\mathrm a})$ is decreasing on $[0,1]$.

Since $\beta_1\beta_3<0$, the roots satisfy
\begin{equation*}
r_{\mathrm{s}} = \frac{-\beta_2\beta_3+\sqrt{\Delta}}{\beta_1\beta_3}
\le
\frac{-\beta_2\beta_3-\sqrt{\Delta}}{\beta_1\beta_3} = r_{\mathrm{l}}.
\end{equation*}
Because the vertex lies {at or} to the right of $[0,1]$, the larger root satisfies $
r_{\mathrm{l}}>1$ and only the smaller root $r_{\mathrm{s}}$ may lie in $[0,1]$. Accordingly, $\bar{J}_{\mathrm{LRA}}^\mathrm{TS}$ is increasing, decreasing, or increasing-then-decreasing on $[0,1]$, depending on whether $r_{\mathrm{s}}\ge 1$, $r_{\mathrm{s}}\le 0$, or $r_{\mathrm{s}}\in(0,1)$, respectively. Therefore,
\begin{equation*}
p_\mathrm{a}^\dagger=
\min\left\{1,\max\left\{\frac{-\beta_2\beta_3+\sqrt{\Delta}}{\beta_1\beta_3},0\right\}\right\}.
\end{equation*}

{It remains to consider $\beta_1\beta_3=0$. If $\beta_3=0$, then $\beta_4=0$, so $F$ is affine. Otherwise, $\beta_1=0$. If additionally $\beta_2=0$, $F\equiv0$; otherwise, its derivative can vanish only at
$
r=-\frac{\beta_4}{2\beta_3}
=\frac{(p_{\mathrm h}^{b}+p_{\mathrm h}^{e})(p+q)}
{2p_{\mathrm h}^{b}p_{\mathrm h}^{e}(p+q-1)}.
$
Since $p_{\mathrm h}^{b}+p_{\mathrm h}^{e}\ge2p_{\mathrm h}^{b}p_{\mathrm h}^{e}$, we have $r<0$ for $p+q<1$ and $r\ge(p+q)/(p+q-1)>1$ for $p+q>1$. Thus, $F$ is monotone on $[0,1]$, and a maximizer is $p_{\mathrm a}^{\dagger}=\mathbbm 1 {\{F(1)>F(0)\}}$.}

\bibliographystyle{IEEEtran}
\bibliography{./bib_files/abrv,./bib_files/ref}

\end{document}